\documentclass[journal]{IEEEtran}
\usepackage{amsmath,amssymb,amsfonts}
\usepackage{amsthm}
\usepackage{calligra}
\usepackage{algorithmicx}
\usepackage[linesnumbered,ruled,vlined]{algorithm2e}
\usepackage{bm}
\usepackage{cite}
\usepackage{hyperref}
\usepackage{enumerate}
\usepackage[table]{xcolor}
\usepackage{graphicx}
\usepackage{subfig}
\usepackage{epstopdf}
\usepackage{cite,url,epsfig}
\usepackage{epstopdf}
\usepackage{caption}

\newtheorem{definition}{Definition}

\newtheorem{lemma}{Lemma}

\begin{document}
\title{Blockchain-empowered Federated Learning for Healthcare Metaverses: User-centric Incentive Mechanism with Optimal Data Freshness}

\author{Jiawen Kang, Jinbo Wen, Dongdong Ye, Bingkun Lai, Tianhao Wu, Zehui Xiong,\\ Jiangtian Nie, Dusit Niyato, \textit{Fellow, IEEE}, Yang Zhang, Shengli Xie, \textit{Fellow, IEEE}

\thanks{

	Jiawen Kang, Jinbo Wen, Dongdong Ye, Bingkun Lai, Tianhao Wu, and Shengli Xie are with the School of Automation, Guangdong University of Technology, China (e-mail: kavinkang@gdut.edu.cn; jinbo1608@163.com; dongdongye8@163.com; bingkunlai@163.com; wutianhao32@163.com; shlxie@gdut.edu.cn).
 
	Zehui Xiong is with the Pillar of Information Systems Technology and Design, Singapore University of Technology and Design, Singapore (e-mail: zehui\_xiong@sutd.edu.sg).
 
	Jiangtian Nie and Dusit Niyato are with the School of Computer Science and Engineering, Nanyang Technological University, Singapore (e-mail: jnie001@e.ntu.edu.sg; dniyato@ntu.edu.sg).
 
	Yang Zhang is with the College of Computer Science and Technology, Nanjing University of Aeronautics and Astronautics, China (e-mail: yangzhang@nuaa.edu.cn).
    
	The work was presented in part at the 2022 IEEE International Conference on Blockchain (Blockchain) (\textit{*Corresponding author: Jiangtian Nie}).
	} 
}

\maketitle

\begin{abstract}
Given the revolutionary role of metaverses, healthcare metaverses are emerging as a transformative force, creating intelligent healthcare systems that offer immersive and personalized services. The healthcare metaverses allow for effective decision-making and data analytics for users. However, there still exist critical challenges in building healthcare metaverses, such as the risk of sensitive data leakage and issues with sensing data security and freshness, as well as concerns around incentivizing data sharing. In this paper, we first design a user-centric privacy-preserving framework based on decentralized Federated Learning (FL) for healthcare metaverses. To further improve the privacy protection of healthcare metaverses, a cross-chain empowered FL framework is utilized to enhance sensing data security. This framework utilizes a hierarchical cross-chain architecture with a main chain and multiple subchains to perform decentralized, privacy-preserving, and secure data training in both virtual and physical spaces. Moreover, we utilize Age of Information (AoI) as an effective data-freshness metric and propose an AoI-based contract theory model under Prospect Theory (PT) to motivate sensing data sharing in a user-centric manner. This model exploits PT to better capture the subjective utility of the service provider. Finally, our numerical results demonstrate the effectiveness of the proposed schemes for healthcare metaverses.

\end{abstract} 

\IEEEpeerreviewmaketitle

\section{Introduction}\label{Intro}

The recent COVID‑19 pandemic has increased the demand for remote healthcare services\cite{GARAVAND2022101029}. With the maturation and applications of metaverse technologies\cite{9880528}, digital healthcare has undergone a revolution, acting as a key force of healthcare industry evolution\cite{kurniasih2022digital}. Unlike traditional videoconferencing-based telemedicine systems, healthcare metaverses are regarded as the future continuum between healthcare industries and metaverses, which blend physical and virtual spaces and break spatial and temporal barriers, providing immersive and interactive healthcare services that meet individual needs of users (e.g., patients or medical staff) \cite{chengoden2022metaverse}. By collaboratively utilizing cutting-edge technologies, such as blockchain\cite{chengoden2022metaverse,xue2023integration}, Federated Learning (FL) \cite{wang2022development}, and digital twins\cite{chengoden2022metaverse}, healthcare metaverses have the potential to cover various applications, such as virtual comparative scanning and metaversed medical intervention \cite{wang2022development}. To build a healthcare metaverse, Internet of Medical Things (IoMT) devices (e.g., wearable devices and embedded medical devices carried by users) play an important role in communication and networking. For example, IoMT devices can collect a large amount of patients' medical data (e.g., temperature, blood pressure, and electrocardiogram) to bridge the physical space and virtual spaces, providing patients with optimal treatment strategies based on the analysis and diagnosis of multiple patients' attributes \cite{chengoden2022metaverse}. 

Although the healthcare metaverse holds great potential for transforming the healthcare ecosystem, this technology still faces many challenges. There are some challenging bottlenecks for future popularization and development: 
1) The healthcare metaverse risks user privacy leaks. Due to privacy concerns, users may be reluctant to share private sensitive data in healthcare metaverses\cite{wang2022development}, which hinders significant data analysis like pharmacodynamic analysis by using Artificial Intelligence (AI) technologies for healthcare improvement. 2) Sensing data suffers from being modified or tampered with by attackers in healthcare metaverses. Since users could have limited power in controlling their data sharing with whom and under what conditions\cite{kostick2022nfts}, the collected data are not safe in healthcare metaverses, and the incorrect or manipulated data being analyzed will cause serious consequences\cite{zhang2023multi}. 3) Due to the energy constraints of IoMT devices, users may not join metaverses or provide fresh data without a proper incentive mechanism. Since the timeliness of healthcare data can affect diagnostic results, fresh sensing data are extremely important to enhance the quality of healthcare metaverse services \cite{9881813}. Therefore, it is necessary to design a user-centric incentive mechanism for incentivizing users with fresh data in healthcare metaverses. Some efforts have been conducted for incentivizing users with data sharing\cite{zhan2020learning,8818983,nie2022blockchain}, but they ignore data freshness and the problem of information asymmetry.

To address the above challenges, in this paper, we first apply FL and cross-chain technologies to design a user-centric privacy-preserving framework for sensing data sharing in healthcare metaverses\cite{9881813}, in which FL technologies can provide privacy protection for users\cite{8832210}, and blockchain technologies can ensure data security for users and efficiently solve the problem of the single point of failure \cite{wenoptimal}. Especially, the blockchain-based healthcare metaverse 
enables users to access any digital space without the involvement of any central institution, which enhances the scalability of the healthcare metaverse\cite{chengoden2022metaverse}. To improve the service quality of healthcare metaverses, we utilize Age of Information (AoI) as a data-freshness metric to quantify sensing data freshness for healthcare metaverse services. Then, with asymmetric information, we design an AoI-based contract model to incentivize fresh data sharing among users. Considering that a service provider (i.e., an FL task publisher) may behave irrationally when facing uncertain and risky circumstances\cite{kahneman2013prospect}, we utilize Prospect Theory (PT) to capture the subjective utility of the service provider, which makes the AoI-based contract model more reliable in practice, and ultimately formulate the subjective utility as the goal function of the model\cite{huang2021efficient,8031035}. The main contributions of this paper are summarized as follows:

\begin{itemize}
    \item We design a new user-centric privacy-preserving framework for healthcare metaverses, where users can keep sensitive sensing data in the physical space for privacy protection and upload non-sensitive sensing data to the virtual space for learning-based metaverse tasks.
    \item To manage sensing data and improve privacy protection, we develop a cross-chain empowered FL framework, which can perform secure, decentralized, and privacy-preserving data training in both virtual and physical spaces through a hierarchical cross-chain architecture consisting of a main chain and multiple subchains. The cross-chain interaction is executed to complete secure model aggregation and updates.
    \item To optimize time-sensitive learning tasks in healthcare metaverses, we apply the AoI as a data-freshness metric of sensing data for healthcare metaverse services and introduce the tradeoff of the AoI and the service latency involving FL-based model training.
    \item We propose an AoI-based contract model under PT to incentivize data sharing among users. To maximize the subjective utility of the service provider subject to necessary constraints, we formulate a PT-based solution for optimal contract design. \textit{To the best of our knowledge, this is the first work to study the data freshness-based incentive mechanism under PT for healthcare metaverses}.
\end{itemize}

The remainder of the paper is organized as follows: In Section \ref{RW}, we review the related work in the literature. In Section \ref{framework}, we propose the cross-chain empowered FL framework for healthcare metaverses. In Section \ref{Problem}, we introduce the AoI and propose the AoI-based contract model under PT. In Section \ref{Optimial_Contract}, we formulate the optimal contract design under PT and propose the corresponding algorithm. Section \ref{Results} presents the security analysis of the proposed framework and numerical results of the proposed incentive mechanism and the framework. Finally, Section \ref{Conclusion} concludes this paper. 

\section{Related Work}\label{RW}
With the help of high-quality immersive content and gamification features, the healthcare metaverse can increase user engagement. For example, it can help clinicians explain complex concepts to patients, provide walk-throughs of the procedures that their patients will receive, and ensure that patients take their prescribed medications accurately\cite{chengoden2022metaverse}.
Given the revolutionary nature of the healthcare metaverse, this technology has been studied recently \cite{bansal2022healthcare, chengoden2022metaverse, ali2023metaverse}. Chengoden \textit{et al.} \cite{chengoden2022metaverse} provided a comprehensive review of the healthcare metaverse, emphasizing state-of-the-art applications, potential projects, and enabling technologies for achieving healthcare metaverses such as FL and blockchain technologies. Bansal \textit{et al.} \cite{bansal2022healthcare} provided a comprehensive survey that examines the latest metaverse development in the healthcare industry, including seven domains such as clinical care, education, and telemedicine. Ali \textit{et al.} \cite{ali2023metaverse} presented the potential of metaverse fusing with AI technologies and
blockchain technologies in the healthcare domain and proposed a metaverse-based healthcare system by integrating blockchains and explainable
AI for the diagnosis and treatment of diseases. Although the healthcare metaverse will revolutionize the healthcare sector, there are foreseeable challenges that require us to solve for the development of the healthcare metaverse, especially privacy and security problems\cite{wang2022development}.

Privacy and security are of critical importance for healthcare metaverses\cite{wang2022development}. To address the privacy concerns of sharing data, FL technologies have been applied for multiple data owners to collaboratively train a global model without sharing their raw data\cite{8994206}. Additionally, relying on encryption technologies and consensus algorithms of distributed systems, blockchain as a distributed ledger technology can effectively solve the problem of security vulnerabilities caused by centralized nodes\cite{wenoptimal}. Since FL technologies can provide privacy protection for users \cite{8994206} and blockchain technologies can ensure the data security of users\cite{zheng2018blockchain}, some works have been conducted for designing a blockchain-empowered FL framework for smart healthcare\cite{chang2021blockchain,jatain2022blockchain,wadhwa2022blockchain}. Chang \textit{et al.} \cite{chang2021blockchain} proposed a blockchain-based FL framework for smart healthcare in which the edge nodes maintain the blockchain to resist a single point of failure and IoMT devices implement the FL to make full of distributed clinical data. Jatain \textit{et al.} \cite{jatain2022blockchain} proposed a blockchain-based FL framework for the secure aggregation of private healthcare data, which can provide an efficient method to train machine learning models. Wadhwa \textit{et al.} \cite{wadhwa2022blockchain} proposed a blockchain-based FL approach for the detection of patients using IoMT devices, which provides security for the detection of patients. However, most works do not consider how to incentivize users to contribute fresh sensing data for reliable healthcare services, especially under information asymmetry. 


\begin{figure*}[t]\centering     \includegraphics[width=0.9\textwidth]{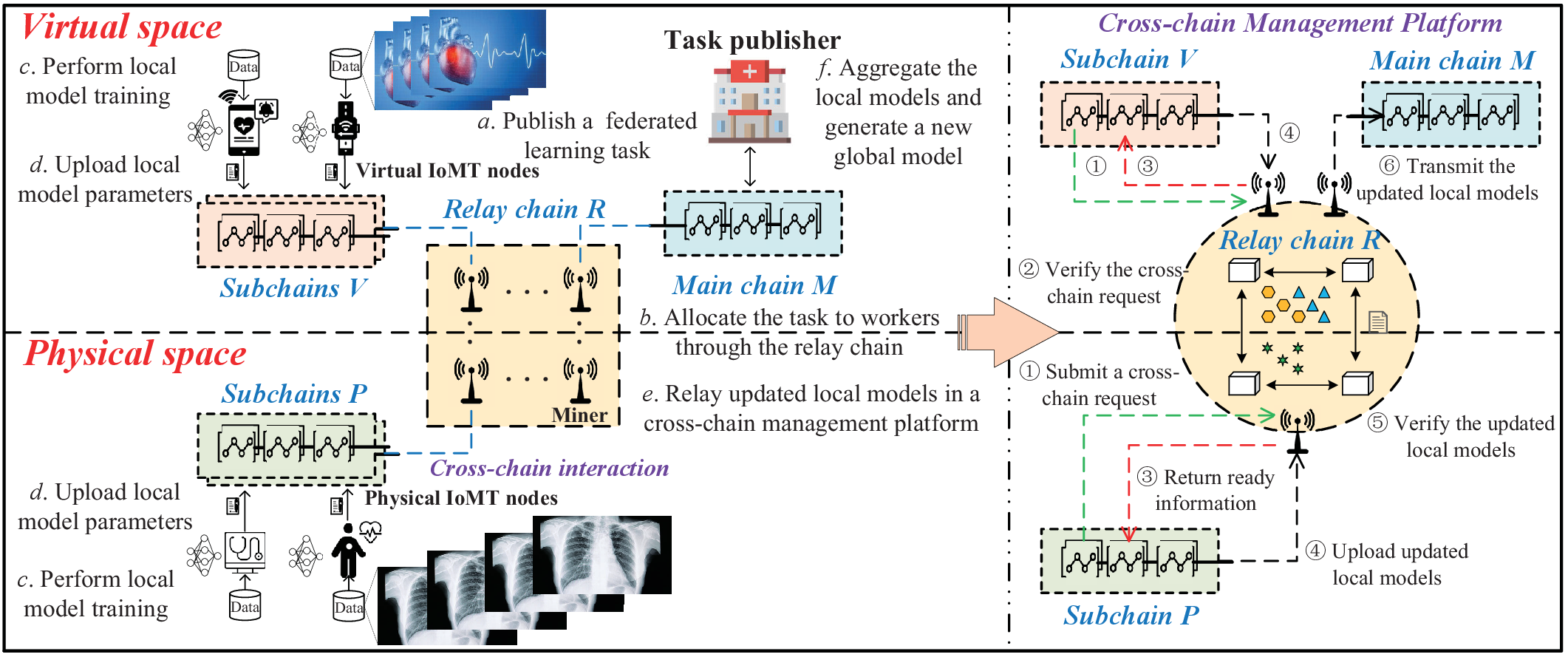}
\captionsetup{font=small}
\caption{A cross-chain empowered FL framework for healthcare metaverses.}    \label{system_model} 
\end{figure*}

To motivate users for sensing data sharing, some efforts have been conducted\cite{zhan2020learning,8818983,nie2022blockchain}. However, most works only consider a complete information scenario and ignore the problem of information asymmetry. Contract theory is a powerful tool for incentive mechanism design under information asymmetry \cite{bolton2004contract}, which has been applied in wireless communication areas\cite{hou2017incentive, Jinbo}. Some works have studied contract-based incentive mechanism design for incentivizing data sharing\cite{zhang2022toward,xuan2020incentive,karimireddy2022mechanisms}, but they ignore the data freshness. Therefore, we focus on designing a contract-based incentive mechanism with optimal data freshness. AoI has been commonly used as an effective metric to
quantify data freshness at the destination. It is defined as the elapsed time from
the generation of the latest received status update, and its minimization depends on the status update frequency \cite{kaul2012real,kosta2017age}. A few works have studied the AoI-based contract model\cite{zhou2021towards,lim2020information}. Zhou \textit{et al.} \cite{zhou2021towards} proposed a contract model considering both the AoI and the service latency to monetize contents in a realistic asymmetric information scenario. Lim \textit{et al.}\cite{lim2020information} proposed a task-aware incentive
scheme based on contract theory that can be calibrated to the model owner's weighted preferences for the AoI and the service latency. However, none of the existing work takes the healthcare metaverse scenario into account. Therefore, it
is urgent to design an AoI-based contract model for healthcare metaverses. Besides, PT considered to be a descriptive model has been widely applied to elucidate how a person's subjective attitude affects decision-making under the uncertainty and risks\cite{8031035,huang2021efficient}. Huang \textit{et al.}\cite{huang2021efficient} formulated the subjective evaluation of offloading users on the utility in computation offloading based on contract theory and PT. Rahi \textit{et al.}\cite{8031035} used PT to account for each prosumer's valuation of its gains and losses. Motivated by the above works, we use PT to capture the subjective utility of the service provider in healthcare metaverses.


\section{Cross-chain Empowered Federated Learning Framework for Healthcare Metaverses}\label{framework}

\subsection{User-centric Privacy-preserving Training Framework}
In the healthcare metaverse, a virtual IoMT node is constructed by mapping and synchronizing the data of a physical IoMT node to the virtual space\cite{9881813}. The virtual space is established on collected data from physical IoMT nodes and online generated data during node interaction and data analysis\cite{9881813}. However, due to privacy concerns, users may not be willing to upload all privacy-sensitive data to the healthcare metaverse directly. Thus, the datasets of virtual nodes are incomplete. If learning tasks are trained by virtual nodes only, the accuracy of learning models will be degraded and the generalization ability will be poor\cite{9881813}. To this end, a user-centric privacy-preserving training framework is designed for healthcare metaverses, where users can customize uploading non-sensitive sensing data to virtual spaces for learning-based metaverse tasks and applications, and keep sensitive sensing data (e.g., heartbeats and chronic conditions) locally in the physical space for strong privacy protection\cite{9881813}.

As shown in Fig. \ref{system_model}, a hierarchical cross-chain architecture for decentralized FL consists of a main chain and multiple subchains\cite{9881813}. This architecture is divided into a physical space and a virtual space. In the physical space, IoMT nodes can be sensors that are integrated into medical systems. They can measure different biological parameters and monitor real-time and healthcare-related data of users\cite{chengoden2022metaverse}, in which the subchains $P$ manage sensitive sensing data and local model updates during training. Similarly, in the virtual space, the main chain $M$ acts as a parameter server that manages global model updates, and the subchains $V$ manage non-sensitive sensing data and local model updates generated by virtual IoMT nodes that act as FL workers\cite{9881813}. 

\subsection{Cross-chain Interaction for Decentralized FL}
To further improve privacy protection of healthcare metaverses, a cross-chain empowered FL framework is further designed to ensure secure, decentralized, and privacy-preserving data training in both virtual and physical spaces via a hierarchical blockchain architecture with the main chain and multiple subchains \cite{kang2022communication}. The cross-chain interaction is executed to complete secure model aggregation and updates, which breaks data islands between the main chain and the subchains\cite{jin2022towards}. As shown in Fig. \ref{system_model}, the workflow of the proposed cross-chain empowered FL framework is presented as follows\cite{9881813}:

\emph{\textbf{Step 1: Publish a federated learning task:}} Each task publisher (e.g., a hospital or a community health center) sets up a learning task (e.g., infectious prediction  of the COVID-19 epidemic) and sends a federated learning request to the main chain $M$ (Step $a$).

\emph{\textbf{Step 2: Allocate the task to workers through the relay chain:}} The main chain $M$ sends the learning task to the relay chain $R$ which is a cross-chain management platform\cite{kang2022communication}. This platform is responsible for verifying data (e.g., model parameters), forwarding data, and bridging connections among the main chain $M$ and the subchains $P$ and $V$\cite{9881813}. The relay chain $R$ first verifies the task information and then forwards the learning task to the workers' subchains $V$ of the virtual space and the workers' subchains $P$ of the physical space, respectively (Step $b$).

\emph{\textbf{Step 3: Perform a learning task in both virtual and physical spaces:}} In the physical space, legitimate IoMT devices (e.g., smart phones, smart watches, and wearable biomedical sensors) can join the training task and perform local model training on their local datasets that involve the sensitive sensing data of users\cite{9881813}. Each physical node trains a given global model locally from the task publisher and updates local model parameters (Step $c$). Similarly, in the virtual space, legitimate virtual nodes that act as FL workers also train the given global model and update local model parameters at the same time (Step $c$).

\emph{\textbf{Step 4: Relay updated local models in the cross-chain management platform:}} When IoMT nodes in the virtual and physical spaces complete the training task, the updated local models are verified and uploaded to their subchains immediately for secure management (Step $d$). To transmit these models to the main chain $M$, the subchains first submit cross-chain requests. When the cross-chain requests are verified successfully by the miners of the relay chain $R$, the relay chain $R$ will return ready information, and the subchains are allowed to upload the updated local models. Finally, the updated local models are checked (e.g., verify that the training times of the models meet the requirement) by the relay chain $R$ and transmitted to the main chain $M$ (Step $e$).

\emph{\textbf{Step 5: Aggregate the local models and update a new global model:}} After transmitting the updated local models to the main chain $M$, the updated local models are aggregated on the main chain $M$ to generate a new global model (Step $f$). Then, the workers download the latest global model from their subchains and train the new global model for the next iteration until satisfying the given accuracy requirement\cite{9881813}. Finally, the final global model is sent back to the task publisher, and the task publisher sends monetary rewards to the workers according to their contributions\cite{jin2021cross}.

Considering that users equipped with energy-limited IoMT devices may be reluctant to contribute fresh sensing data for time-sensitive FL tasks in healthcare metaverses, a reliable incentive can be used to encourage users to share fresh sensing data, which is discussed in Section \ref{Problem}.

\section{Problem Formulation}\label{Problem}

In this section, to incentivize data sharing among users for time-sensitive FL tasks, we first introduce the AoI as an effective metric to evaluate sensing data freshness. Then, we formulate the utility functions of both workers and the service provider in healthcare metaverses.

Referring to \cite{9881813}, we consider a mixed reality based remote monitoring as an example of healthcare metaverse scenarios with a service provider and a set $\mathcal{M} = \small\{1,\ldots,m,\ldots,M\small\}$ of $M$ workers. The service provider acting as the task publisher motivates $M$ workers to participate in learning tasks. The average time of a global model iteration in the cross-chain empowered FL consists of three parts: 1) The average time of completing a global model iteration of FL (denoted as $t_u$); 2) The average time of completing a consensus process for a global model iteration among blockchains (denoted as $t_c$); 3) The average time of collecting and processing the data for model training (denoted as $c_m t,\:c_m \in \mathbb{N}$, and $t = t_u + t_c$). Considering that the FL is synchronous, $t_u$ is the same for all the workers \cite{lim2020information}, and $t_c$ is the same for each global model iteration because of using the same relay chain \cite{shen2020blockchain}. 
For each worker $m$, the time of collecting and processing data for model training is a constant \cite{lim2020information}. 


\subsection{AoI and Service Latency for Healthcare Metaverses}
AoI has been a well-accepted metric to quantify data freshness and improve performances of time-critical applications and services, especially for sensor networks\cite{lim2020information}. 
In this paper, we define AoI as the time elapsed from the beginning of data aggregation by deployed IoMT devices to the completion of FL-based training, and the service latency as the time elapsed from the initiation of the FL training request to the completion of FL-based training. We focus on the AoI and the service latency of FL with a data caching buffer in workers \cite{lim2020information}, where the low AoI determines the high quality of FL-based training for reliable healthcare metaverse services. Without loss of generality, we consider that a model training request arrives at the beginning of each epoch\cite{9881813}.
Worker $m$ periodically updates its cached data, and the periodic interval $\theta_m$ is independent of the period in which the request arrives, which is denoted as
\begin{equation}
    \begin{split}
        \theta_m = c_m t + a_m t,\: a_m \in \mathbb{N},
    \end{split}
\end{equation}
where $(a_mt)$ is the duration from finishing data collection to the beginning of the next phase of  data collection. Note that $(a_m t)$ can be service time or idle time in terms of multiple time periods\cite{lim2020information}.


Following the characteristic of the Poisson process, the probability of a request's arrival is identical across periods\cite{lim2020information}. If an FL training request is raised at the $z$-$th$ period during the data collection phase, the service latency is $c_m t + t - (z-1) t$. 
Otherwise, if a request is raised at any remaining time period in the update cycle, the service latency is $t$. 
Thus, the average service latency $D_m$ \cite{lim2020information} of the blockchain-based FL for worker $m$ is given by
\begin{equation} \label{eqation:D_s}
    \begin{split}
        \overline{D}_m = \frac{c_m}{ c_m + a_m } \left[ \frac{c_m t}{2} (c_m + 3) \right]  + \frac{a_m t}{ c_m + a_m }. \\
    \end{split}
\end{equation}

For content caching, if an FL request is raised during the data collection phase or at the beginning of phase $(c_m + 1)t$, the AoI is $t$ which is the minimum value. 
Otherwise, if a request is raised at period $(lt)$, the AoI is $[l - (c_m + 1) + 1]t$, where $l \geq (c_m + 2)t$. 
Thus, the average AoI\cite{lim2020information} for worker $m$ is given by
\begin{equation}\label{eqation:A_s}
    \begin{split}
        \overline{A}_m = \frac{t}{ c_m + a_m } \left[ c_m + 1 + \frac{(a_m - 1)(a_m + 2) }{2}   \right].  \\
    \end{split}
\end{equation}

When $t$ is fixed, the update cycle $\theta_m$ is affected by $c_m$ and $a_m$. Therefore, we consider two general cases:


\emph{\textbf{Case 1: Adjustable Idle Phase and Fixed Update Phase: }}
When $c_m = c,\:c\in \mathbb{N}$ is fixed, we have $a_m = \frac{\theta_m }{t} - c $. Replacing $a_m$ with $\theta_m$, $\overline{D}_m$ and $\overline{A}_m$ can be simplified to\cite{9881813}
\begin{equation} \label{eqation:D_s_1}
\begin{split}
\overline{D}_m(\theta_m) & = \frac{c_m}{ c_m + a_m } \left[ \frac{c_m t}{2} (c_m + 3) \right]  + \frac{a_m t}{ c_m + a_m }  \\
& = \frac{c^2 t^2 (c+3) }{2 \theta_m} + \frac{(\theta_m - ct)t }{ \theta_m }\\
& =  \frac{ct^2 (c^2 + 3c -2) }{2 \theta_m}   + 1\\
\end{split}
\end{equation}
and
\begin{equation}\label{eqation:A_s_1}
\begin{split}
\overline{A}_m(\theta_m) & = \frac{t}{ c_m + a_m } \left[ c_m + 1 + \frac{(a_m - 1)(a_m + 2) }{2}   \right]  \\
& = \frac{t^2}{\theta_m} \left[ c + 1 + \frac{  (\theta_m - ct - t)( \theta_m - ct + 2t ) }{2 t^2}    \right]  \\
& =  \frac{\theta_m}{2} + \frac{t - 2ct }{2} + \frac{ c^2 t^2 + ct^2 }{ 2 \theta_m} . \\
\end{split}
\end{equation}
Since $\frac{\mathrm{d}^2\overline{A}_m(\theta_m)}{\mathrm{d}\theta_m^2} = \frac{c^2t^2+ct^2}{\theta_m^3}>0$, $\overline{A}_m(\theta_m) $ is a convex function with respect to $\theta_m$. Besides, when $c \geq 1,\: c \in \mathbb{N}$, we have $c^2 + 3c - 2 > 0$. Thus, $\overline{D}_m(\theta_m)$ is a convex function with respect to $\theta_m$. With the update cycle $\theta_m$ decreasing, the service latency increases while the AoI may decrease. 
In other words, we can adjust the update cycle $\theta_m$ to tradeoff the AoI and the service latency\cite{lim2020information,zhang2018towards}. 

\emph{\textbf{Case 2: Adjustable Update Phase and Fixed Idle Phase: }}
When $a_m = a,\:a\in \mathbb{N}$ is fixed, we have $c_m = \frac{\theta_m }{t} - a $. Replacing $c_m$ with $\theta_m$, $\overline{D}_m$ and $\overline{A}_m$ can be simplified to\cite{9881813}
\begin{equation} \label{eqation:D_s_2}
\begin{split}
\overline{D}_m(\theta_m) & = \frac{c_m}{ c_m + a_m } \left[ \frac{c_m t}{2} (c_m + 3) \right]  + \frac{a_m t}{ c_m + a_m }  \\
& = \frac{(\theta_m - at)^3}{2t\theta_m} + \frac{3 (\theta_m - at)^2}{2\theta_m} + \frac{at^2}{\theta_m}  \\
\end{split}
\end{equation}
and
\begin{equation}\label{eqation:A_s_2}
\begin{split}
\overline{A}_m(\theta_m) & = \frac{t}{ c_m + a_m } \left[ c_m + 1 + \frac{(a_m - 1)(a_m + 2) }{2}   \right]  \\
& =  \frac{t \theta_m }{\theta_m -at} + \frac{t^2}{\theta_m -at} \left(  \frac{a^2 - a }{2} \right).  \\
\end{split}
\end{equation}
Since $\theta_m = a_m t + c_m t$, $\theta_m > a_m t $ always holds. 
When $\theta_m > at $, we have $\frac{\mathrm{d}^2\overline{D}_m(\theta_m)}{\mathrm{d}\theta_m^2} = \frac{\theta_m^3-at^3(a^2-3a+2)}{t\theta_m^3}>0$. Thus, $\overline{D}_m(\theta_m)$ is a convex function with respect to $\theta_m$.
When $\theta_m > at $ and $a>1$, we have $\frac{\mathrm{d}^2\overline{A}_m(\theta_m)}{\mathrm{d}\theta_m^2} = \frac{at^2(a-1)}{(\theta_m-at)^3}>0$. Thus, $\overline{A}_m(\theta_m)$ is also a convex function with respect to $\theta_m$. 

\subsection{Worker Utility}

The utility of worker $m$ is the difference between the received monetary reward $R_m$ and its cost $C_m$ of participating in FL training tasks, which is presented as $U_m = R_m - C_m$.
Referring to \cite{zhou2021towards}, we have $C_m = \delta_m/ \theta_m$,
where $\delta_m$ is the update cost per time and is related to data collection, computation, transmission, and consensus \cite{zhang2018towards, zhou2021towards}. Thus, the utility of worker $m$ is rewritten as 
\begin{equation}
    \begin{split}
        U_m = R_m - \frac{\delta_m}{ \theta_m}.  \\
    \end{split}
\end{equation}

Due to information asymmetry, the service provider is not aware of the update cost of each worker precisely, but it can sort the workers into discrete types by using the statistical distributions of worker types from historical data to optimize the expected utility of the service provider\cite{hou2017incentive}. Specifically, we divide the workers into different types and denote the $n$-th type worker as $\delta_n$.
The workers can be classified into a set $\mathcal{N}= \left\{ \delta_n:1 \leq n \leq N \right\}$ of $N$ types.
In non-decreasing order, the worker types are sorted as
$\delta_1 \geq \delta_2 \geq \dots \geq \delta_N$\cite{9881813}. To facilitate explanation, the worker with type $n$ is called the type-$n$ worker.
Thus, the utility of the type-$n$ worker can be rewritten as 
\begin{equation}
        U_n = R_n - \frac{\delta_n }{ \theta_n}.
\end{equation}
To simplify the description, we define the update frequency as $f_n = \frac{1}{\theta_n}$. The worker type is redefined as $\gamma_n = \frac{1}{\delta_n}$ and the worker types, i.e., $\delta_1 \geq \delta_2 \geq \dots \geq \delta_N$, are rewritten as $\gamma_1 \leq \gamma_2 \leq \dots \leq \gamma_N$\cite{9881813}. Thus, the utility of the type-$n$ worker can be rewritten as
\begin{equation}
    \begin{split}
        U_n = R_n - \frac{f_n }{ \gamma_n}.  \\
    \end{split}
\end{equation}

\subsection{Prospect Theory}
According to the conventional decision theory, the service provider is always rational and optimizes the decision-making process to maximize its own utility based on Expected Utility Theory (EUT), which uses objective probabilities to determine the weight of each possible payoff \cite{tang2018multi}. 
However, in an uncertain and risky environment, the service provider may behave irrationally and prefer to adjust original decisions in a predefined manner. Therefore, EUT is not applicable to capture risk attitudes of the service provider during the uncertain decision-making process. In the next subsection, we use both EUT and PT to capture the utility of the service provider. Firstly, we introduce the effect of two key notions from PT, i.e., \textit{probability weighting} and \textit{utility framing}.

\emph{1) Probability weighting effect:} Different from EUT, PT uses a subjective probability to determine the weight of each possible payoff. The subjective probability is a function in terms of the objective probability, which illustrates that high probability events are underestimated and low probability events are overestimated\cite{ye2022incentivizing, kahneman2013prospect}. 

\emph{2) Utility framing effect:} PT utilizes a reference point to frame the payoff of each outcome into either gain or loss. For instance, the service provider defines the goal of earning a specific amount of profits as its reference point. If its goal is not reached, it will perceive that it is a non-positive loss.
In summary, the utility of EUT is given by $U_{\text{EUT}}=\sum_{n=1}^N Q_n U_{n,\text{EUT}}$,
where $Q_n$ is the objective probability and $U_{n,\text{EUT}}$ is the outcome for the alternative $n$. Following the probability weighting effect and the utility framing effect, the utility of PT is defined as $U_{\text{PT}}=\sum_{n=1}^N H(Q_n) U_{n,\text{PT}}$,
where $H (\cdot)$ is an inverse S-shaped probability weighting function in terms of the objective probability $Q$. Referring to \cite{kahneman2013prospect}, the probability weighting function is denoted as $H =\exp(-(-\log(Q))^{\rho}) $, where $\rho$ is a rational coefficient that reveals how a person's subjective evaluation distorts objective probabilities. The more rational players have a higher $\rho$, while the more subjective players have a lower $\rho$. Thus, $U_{n,\text{PT}}$ is defined as\cite{8031035,huang2021efficient,ye2022incentivizing} 
\begin{equation} \label{PT}
U_{n,\text{PT}}= \left\{ \begin{aligned}
(U_{n,\text{EUT}}-U_{n, \text{ref}})^{\zeta^{+}},\: U_{n,\text{EUT}}\geq U_{n, \text{ref}},\\
- \eta (U_{n,\text{ref}}-U_{n,\text{EUT}})^{\zeta^{-}},\: U_{n,\text{EUT}}<U_{n,\text{ref}},\\
\end{aligned} \right.
\end{equation}
where $\zeta^{+},\: \zeta^{-} \in (0,1]$ are two weighting factors that formulate the gain and loss distortions, respectively. $\eta \geq 0$  is a loss aversion coefficient. $U_{n, \text{ref}}$ is a reference point framing the utility of $U_{n,\text{EUT}}$ into either gain or loss.

\subsection{Service Provider Utility}
Since large AoI and large service latency lead to a bad immersive experience for users and reduce the satisfaction of the service provider in healthcare metaverses \cite{jiang2021reliable}, the satisfaction function of the service provider obtained from the type-$n$ worker is defined as\cite{zhou2021towards}
\begin{equation}
    \begin{split}
        G_n = \beta g(f_n),
    \end{split}
\end{equation}
where $\beta >0$ is the unit profit for the performance and $g(\cdot)$ is the performance obtained from the type-$n$ worker, which is defined as \cite{zhou2021towards}
\begin{equation}\label{eqation:g_s}
    \begin{split}
        g(f_n) = \alpha_n (K - \overline{A}_n) + (1 - \alpha_n) ( H - \overline{D}_n ),
    \end{split}
\end{equation}
where $\alpha_n \in [0,1]$ represents the preference of AoI over service latency for the service provider to the type-$n$ worker, i.e., the larger $\alpha_n$ means that the service provider prefers the AoI more, and $K$ and $H$ are the maximum tolerant AoI and the maximum tolerant service latency, respectively.

Because of information asymmetry, the service provider only knows the number of workers and the type distribution but cannot know the private type of each worker\cite{Jinbo}, namely the exact number of workers belonging to each type, which results in the uncertainty when the service provider makes decisions. Therefore, the service provider overcomes the information asymmetry problem by adopting EUT to define its own objective utility as\cite{9881813}
\begin{equation}\label{F_EUT}
\begin{split}
U_{s, \text{EUT}}  ={\sum^N_{n=1}} M Q_{n} U_{s, n, \text{EUT}} ,\\
\end{split}
\end{equation} 
where $U_{s,n, \text{EUT}} = U_{s,n} = (G_n  - R_n)$ is the objective utility gained from type-$n$ workers and $Q_{n}$ is the probability that a worker is type-$n$. Note that $\sum_{ n=1 }^N Q_{n} = 1$. 

However, when facing uncertain and risky circumstances, the service provider may behave irrationally and have different risk attitudes. Therefore, EUT is not applicable to capture risk attitudes of the service provider during the uncertain decision-making process. In this paper, we utilize PT to further capture the utility of the service provider, which makes the contract model more acceptable in practice.
Given a reference point $U_{\text{ref}}$ for all types of workers, we convert $U_{s, n, \text{EUT}}$ into the subjective utility, which is given by\cite{huang2021efficient, ye2022incentivizing} 
\begin{equation} \label{PT}
U_{s,n,\text{PT}}= \left\{ \begin{aligned}
( U_{s,n, \text{EUT}}- U_{\text{ref}}  )^{\zeta^{+}}, \: U_{s,n, \text{EUT}}\geq U_{\text{ref}},\\
- \eta (  U_{\text{ref}} - U_{s,n, \text{EUT}}   )^{\zeta^{-}},\: U_{s,n, \text{EUT}}< U_{\text{ref}}.\\
\end{aligned} \right.
\end{equation}
Based on (\ref{PT}), the subjective utility of the service provider is presented as
\begin{equation}\label{F_PT}
\begin{split}
U_{s,\text{PT}} = {\sum^N_{n=1}} M Q_{n} U_{s,n, \text{PT}} .\\
\end{split}
\end{equation}

\subsection{Contract Formulation}

The types of workers are private information that is not visible to the service provider, namely there exists information asymmetry between the service provider and the workers. Since contract theory is a powerful tool for designing incentive mechanisms with asymmetric information\cite{hou2017incentive,Jinbo}, the service provider uses contract theory to effectively motivate workers to contribute sensing data for time-sensitive FL tasks.
Here, the service provider is the leader in designing a contract with a group of contract items, and each worker selects the best contract item according to its type. 
The contract item can be denoted as $\Phi = \left\{ (f_{n}, R_{n}), n \in \mathcal{N}   \right\}$, where $f_n$ is the update frequency of the type-$n$ worker and $R_{n}$ is the reward paid to the type-$n$ worker as an incentive for the corresponding contribution\cite{9881813,zhou2021towards}. To ensure that each worker automatically chooses the contract item designed for its specific type, the feasible contract must satisfy the following Individual Rationality (IR) and Incentive Compatibility (IC) constraints\cite{hou2017incentive}.

\begin{definition}
(Individual Rationality) The contract item that
a worker should ensure a non-negative utility, i.e.,
\begin{equation}\label{IR2}
    \begin{split}
        R_{n} - \frac{  f_n }{ \gamma_{n} }  \geq 0, \:\forall n \in \mathcal{N}.
    \end{split}
\end{equation}
\end{definition}

\begin{definition}
(Incentive Compatibility) A worker of any type
$n$ prefers to select the contract item $(f_{n}, R_{n})$ designed for
its type rather than any other contract item $(f_{i}, R_{i}),\: i\in \mathcal{N}$, and $i\neq n$, i.e.,
\begin{equation}
    \begin{split}\label{IC1}
        R_{n} - \frac{  f_n }{ \gamma_{n} }  \geq  R_{i} - \frac{  f_i }{ \gamma_{n} }  , \:\forall n, i \in \mathcal{N},\: n\neq i.
    \end{split}
\end{equation}
\end{definition}

With the IR and IC constraints, the problem of maximizing the expected utility of the service provider is formulated as
\begin{equation}\label{problem1}
    \begin{split}
        \textbf{Problem 1:} & \quad \max_{ \bm{f}, \bm{R}} U_{s, \text{PT}} \\
        \text{s.t.} & \quad \text{Constraints in}\: (\ref{IR2})\: \text{and}\: (\ref{IC1}),\\
        & \quad f_n \geq 0, R_n \geq 0, \gamma_n > 0,\: \forall n \in \mathcal{N},
    \end{split}
\end{equation}
where $\bm{f} =[f_{n}]_{1 \times N}$ and $\bm{R} =[R_{n}]_{1 \times N}$. 

\section{Optimal Contract Design under Prospect Theory}\label{Optimial_Contract}
Since there are $N$ IR constraints and $N(N-1)$ IC constraints in \textbf{Problem 1}, it is difficult to directly solve \textbf{Problem 1} with complicated constraints. Thus, we first reduce the number of attached constraints to reformulate \textbf{Problem 1}. Then, we further derive the EUT-based solution and the PT-based solution theoretically.


\subsection{Contract Reformulation with Reduced Constraints}
\begin{lemma}\label{lemma:1}
    With asymmetric information, a feasible contract must satisfy the following conditions:
    \begin{subequations}
        \begin{align}
            &R_1 -  \frac{f_1}{\gamma_1} \geq 0,\label{IR_lem1}\\
            & 0 \leq f_1 \leq f_2 \leq \cdots \leq f_N,\:0 \leq R_1 \leq R_2 \leq \cdots \leq R_N, \label{IC3_lem1}\\
            &R_n -  \frac{f_n}{\gamma_n} \geq R_{n-1} -  \frac{f_{n-1}}{\gamma_{n}} ,\:\forall n \in \left\{2,\ldots,N\right\},\label{IC1_lem1}\\
            &R_n -  \frac{f_n}{\gamma_n} \geq R_{n+1} -  \frac{f_{n+1}}{\gamma_{n}} ,\:\forall n \in \left\{1,\ldots,N-1\right\}.\label{IC2_lem1}
        \end{align}
    \end{subequations}
\end{lemma}
\begin{proof}
Please refer to \cite{hou2017incentive}. 
\end{proof}

Constraint (\ref{IR_lem1}) related to the IR constraints ensures that the utility of each worker receiving the contract item of its type is non-negative. Constraints (\ref{IC3_lem1}), (\ref{IC1_lem1}), and (\ref{IC2_lem1}) are related to the IC constraints. Specifically, constraint (\ref{IC3_lem1}) indicates that a worker type with a lower cost can provide the service provider with a higher update frequency. Constraints (\ref{IC1_lem1}) and (\ref{IC2_lem1}) show that the IC constraints can be reduced as local downward incentive compatibility and local upward incentive compatibility with monotonicity, respectively\cite{hou2017incentive}. From \textbf{Lemma \ref{lemma:1}}, we can know that when the lowest-type workers satisfy the IR constraints, the other types of workers will automatically hold the IR constraints. When type-$n$ and type-$(n-1)$ workers satisfy the IC constraints, the type-$n$ and the other types of workers will automatically hold the IC constraints. Therefore, the original $(N + N(N-1) = N^2)$ IR and IC constraints are transformed into $(N + 1)$ constraints, and \textbf{Problem 1} can be reformulated as
\begin{equation}\label{P2}
\begin{split} 
   \textbf{Problem 2:} &\: \max_{\small\{(f_n,R_n)\small\}}  U_{s, \text{PT}}\\
    \text{s.t.} &\quad R_1 -  \frac{f_1}{\gamma_1} = 0,\\
    &\quad R_n -  \frac{f_n}{\gamma_n} = R_{n-1} -  \frac{f_{n-1}}{\gamma_{n}} ,\:\forall n \in \left\{2,\ldots,N \right\},\\ 
    &\quad f_n \geq f_{\text{min}}, \: R_n \geq 0, \: \forall n \in \left\{1,\ldots,N\right\},
\end{split}
\end{equation}
where $f_{\text{min}}$ is the minimum update frequency.

Based on the first two constraints of (\ref{P2}), the optimal reward $R_n^*$ can be calculated by the iterative method in a subsequent way, which is given by $R_n^* = \frac{ f_1 }{\gamma_1}+\sum_{i=1 }^n\Delta_i,\: n\in \mathcal{N}$,
where $\Delta_1 = 0$ and $\Delta_i = \frac{f_i}{\gamma_i}-\frac{f_{i-1}}{\gamma_i},\: i = 2,\ldots,N$. Therefore, we can obtain $R_n^*$ as
\begin{equation}\label{pi}
R_n^* = \left\{
\begin{aligned}
&\: \frac{f_n}{\gamma_n} + \sum_{i=1}^{n-1}\bigg(\frac{f_i}{\gamma_i}-\frac{f_i}{\gamma_{i+1}}\bigg), \:2\leq n \leq N,\\
&\:\frac{f_1}{\gamma_1}, \:n = 1.
\end{aligned}
\right.
\end{equation}

\subsection{EUT-based Solution}
According to \cite{chen2018efficient}, we consider a special case that all the worker types are uniformly distributed across all types $N$ in the EUT and PT solutions. Thus, the ratio of each worker type is identical, i.e., $Q_n = Q = 1/N$. Since the service provider is unable to know various ratios of different worker types in the initial phase, it may prefer to temporarily consider the worker types following the uniform distribution and design the contract-based incentive mechanism under this assumption\cite{9881813,Jinbo}.
By substituting (\ref{pi}) into (\ref{F_EUT}), the objective utility of the service provider in terms of $f_n$ is reformulated as
\begin{equation}
\begin{aligned}
    U_{s, \text{EUT}} (f_n) = MQ\sum_{n =1}^{N} \Big( G_n(f_n)  -  b_n f_n\Big),
\end{aligned}
\end{equation}
where 
    \begin{equation}\label{b_n}
         b_n = \left\{
        \begin{split}
            & \frac{1}{\gamma_n}+\left(\frac{1}{{\gamma}_{n}}-\frac{1}{{\gamma}_{{n}+1}}\right)(N-n),\: n < N,\\
            & \frac{ 1 }{{\gamma}_{n}},\:n = N.
        \end{split}
        \right.
    \end{equation}

To maximize $U_{s, \text{EUT}}$, we use the first-order optimality condition $\partial U_{s, \text{EUT}} /\partial f_n = 0$ and obtain $\hat{f}_{n,\text{EUT}}^{*}$. We simultaneously consider the lower bound of the update frequency $f_{ \text{min} }$ to derive the EUT-based solution of $f_n$ as follows:
\begin{equation}\label{f_k,EUT_min}
    \begin{aligned}
    f_{n,\text{EUT}}^{*} = \max \Big(\hat{f}_{n,\text{EUT}}^{*}, f_{ \text{min} } \Big).
    \end{aligned}
\end{equation}

\begin{lemma}\label{lemma:2}
  If $\gamma_1 <\cdots < \gamma_n < \cdots < \gamma_N$, then $U_{1} < \cdots < U_{n} < \cdots < U_{N}$.
\end{lemma}

\begin{proof}
Please refer to Appendix A of \cite{huang2021efficient}.
\end{proof}

\begin{lemma}\label{lemma:3}
  If $\gamma_1 < \cdots < \gamma_n < \cdots < \gamma_N$ and $\frac{1}{\gamma_{n-1}}+\frac{1}{\gamma_{n+1}}-\frac{2}{\gamma_n} \geq 0,\: 1 < n < N$, then $b_1 > \cdots > b_n > \cdots > b_N$.
\end{lemma}

\begin{proof}
Please refer to Appendix B of \cite{huang2021efficient}.
\end{proof}

\begin{lemma}\label{lemma:4}
If $Q_n = Q,\: \forall n$ and $b_1 > \cdots > b_n > \cdots > b_N$, then $U_{s,1,\text{EUT}} \leq \cdots \leq U_{s,n,\text{EUT}} \leq \cdots \leq U_{s,N,\text{EUT}}$.
\end{lemma}

\begin{proof}
Please refer to Appendix C of \cite{huang2021efficient}.
\end{proof}

\textbf{Lemma \ref{lemma:2}} states that as the type of workers increases, the utilities of workers corresponding to each type also increase. \textbf{Lemma \ref{lemma:3}} and \textbf{Lemma \ref{lemma:4}} indicate that as the type of workers increases, the objective utility of the service provider gained from each type of workers also increases, where \textbf{Lemma \ref{lemma:4}} is derived based on \textbf{Lemma \ref{lemma:3}}.

\subsection{PT-based Solution}
Since the EUT-based solution is unable to capture the psychological behavior of the service provider, the EUT-based solution may be suboptimal from the perspective of the service provider. Therefore, we derive the PT-based solution to maximize the subjective utility of the service provider. We consider a special condition $\xi^{+} = \xi^{-} = 1$ to obtain a closed-form PT-based solution, which has been introduced in \cite{8031035}. Based on \textbf{Lemma \ref{lemma:4}}, the PT-based solution of $f_n$ is discussed in the following three cases.

\textit{Case 1:} When $U_{s,n,\text{PT}} \geq U_{\text{ref}}$, the subjective utility in (\ref{F_PT}) can be simplified to 
\begin{equation}\label{V_PT}
    \begin{aligned}
    U_{s,\text{PT}} = MQ\sum_{n = 1}^N  (U_{s,n, \text{PT}} - U_{\text{ref}}).
    \end{aligned}
\end{equation}
By substituting (\ref{pi}) into (\ref{V_PT}), $U_{s,\text{PT}}$ is converted into
\begin{equation}
    \begin{aligned}
    U_{s, \text{PT} }(f_n) = MQ\Bigg(\sum_{n =1}^N  G_n(f_n) - \sum_{n=1}^N b_n f_n - \sum_{n=1}^N  U_{\text{ref}}\Bigg).
    \end{aligned}
\end{equation}
According to the analysis of \emph{\textbf{Case 1}} and \emph{\textbf{Case 2}} in Section \ref{Problem}, $\partial^2 U_{s,\text{PT} }/\partial f_n^2  < 0$ holds when $c \geq 1$. Therefore, we use the first-order optimality condition $\partial U_{s, \text{PT}}/\partial f_n = 0$ to obtain the PT-based solution of $f_n$, which is given by
\begin{equation}
    \begin{aligned}
        f_{n,\text{PT}}^{*} = f_{n,\text{EUT}}^{*}.
    \end{aligned}
\end{equation}
According to \textbf{Lemma \ref{lemma:4}}, when $U_{s, 1,\text{PT}} = U_{s,1,\text{EUT}}  \geq U_{\text{ref}}$, \textit{Case 1} is satisfied.

\textit{Case 2:} When $U_{s,n,\text{PT}} < U_{\text{ref}},\:\forall n$, the subjective utility in (\ref{F_PT}) can be simplified to
\begin{equation}
    \begin{aligned}
    U_{s,\text{PT}} = \eta MQ \sum_{n = 1}^N  (U_{s,n, \text{PT}} - U_{\text{ref}}).
    \end{aligned}
\end{equation}
In \textit{Case 2}, the PT-based solution of $f_n$ is the same as the EUT-based solution of $f_n$. Based on \textbf{Lemma \ref{lemma:4}}, when $U_{s,N,\text{PT}} = U_{s,N,\text{EUT}} \leq U_{\text{ref}}$, \textit{Case 2} is satisfied. Due to the worker types following the uniform distribution, the PT-based solutions of $f_n$ in these two cases are identical. Besides, they are also equal to the EUT-based solution of $f_n$ because of the special condition $\xi^{+} = \xi^{-} = 1$.

\begin{algorithm}[t]
    \caption{Optimal Contract Design under PT}\label{AI_Contract}

    \KwIn{Initialize parameters $\small\{ \alpha, \beta, K, H, t, U_{\text{ref}}\small\}$ and worker types $\small\{\gamma_n, 1 \leq n \leq N\small\}$.}
    \KwOut{$\small\{f_{n,\text{PT}}^{*},R_n^*,\:1\leq n \leq N \small\}$.}

    Derive the EUT-based solution ${f_{n,\text{EUT}}^{*}}$ by (\ref{f_k,EUT_min}).\\
    Calculate $U_{s,1,\text{EUT}}$ and $U_{s,N,\text{EUT}}$.\\
    \If{ $U_{s,1,\text{EUT}} \geq U_{\rm{ref}}$ or $U_{s,N,\text{EUT}} \leq U_{\rm{ref}}$}
    {
        Let $f_{n,\text{PT}}^{*} = f_{n,\text{EUT}}^{*},\:\forall n$.
    }
    \Else{
        \If{$\eta == 1$}{$f_{n,\text{PT}}^{*} = f_{n,\text{EUT}}^{*},\:\forall n$.}
        \Else{
        \If{$\eta < 1$}{Find an index $m$ that satisfies $U_{s,m+1,\text{EUT}} \geq U_{\text{ref}} \geq U_{s,m,\text{EUT}}$.\\
        Calculate $f_{n,\text{PT}}^{*},\:1\leq n \leq m$ based on (\ref{f_a,PT_min}), and let $f_{n,\text{PT}}^{*} = f_{n,\text{EUT}}^{*},\: m+1 \leq n \leq N$.}
        \Else{Find the above $m$ and acquire the feasible solution $f_{n,\text{PT}}^{*}$ by using a polling method and the scheme in \cite{gao2011spectrum}.}
        \textbf{end}

        }
        \textbf{end}

    }
    \textbf{end}
    
    For each worker type $n$, substitute $f_{n,\text{PT}}^{*}$ into (\ref{pi}) to calculate $R_n^*$.\\
    \textbf{return} $\small\{f_{n,\text{PT}}^{*},R_n^*,\:1\leq n \leq N \small\}$.
\end{algorithm}

\begin{algorithm}[t]
\caption{Training Process for Vertical FL}\label{Al_FL}

\KwIn{Training dataset $\boldsymbol{X}^{\textrm{train}}$\ and an initial global model $\{ \boldsymbol{w}_{\text A},\boldsymbol{w}_{\text B}\}$.}
\KwOut{Final global model $\{ \boldsymbol{w}_{\text A}^{\textrm{*}},\boldsymbol{w}_{\text B}^{\textrm{*}}\}$.}

\For{\normalfont{each epoch} $t=1,2,\ldots,T$}{
    \For{\normalfont{each batch of training data}           $\boldsymbol{x}^{\text{batch}}\in\boldsymbol{X}^{\textrm{train}}$}{
        Split the data into two parts:$\{\boldsymbol{x}_\text{A}^{\textrm{batch}},\boldsymbol{x}_\text{B}^{\textrm{batch}}\}$.
        
        Server sends public keys to clients A and B.
        
        Client A sends intermediate information encrypted by using the homomorphic encryption algorithm \textit{Paillier}\cite{paillier1999public} to client B: 
        B$ \gets  \phi_{\text A}(\boldsymbol{x}_\text{A}^{\textrm{batch}},\boldsymbol{w}_{\text A})$.
        
        Client B sends encrypted intermediate information \cite{paillier1999public} to client A: 
        A$ \gets \phi_{\text B}(\boldsymbol{x}_B^{\textrm{batch}},\boldsymbol{w}_{\text B})$.
        
        Clients A and B compute the encrypted local gradient, respectively.
        
        Each client adds a random number and sends the encrypted gradient to the server:
        Server$\gets (\frac{\partial l}{\partial \boldsymbol{w}_{\text A}} + {D}_{\text A})$,
        Server$\gets ( \frac{\partial l}{\partial \boldsymbol{w}_{\text B}} + {D}_{\text B})$.
        
        Server decrypts the gradients and sends them back to clients A and B.
        
        Update local models $\{ \boldsymbol{w}_{\text A},\boldsymbol{w}_{\text B}\}$.
    }
}
\end{algorithm}

\begin{table}[t]\label{parameter}
	\renewcommand{\arraystretch}{1}
        \captionsetup{font = small}
	\caption{ Key Parameters in the Simulation. }\label{table} \centering 
	\begin{tabular}{m{4.7cm}<{\raggedright}|m{2.7cm}<{\centering}}	 	
		\hline		
		\textbf{Parameters} & \textbf{Setting}\\	
		\hline
		Time taken for completing a global iteration and the consensus process $t$ &  $2\: \rm{s}$\\	
		\hline
		Unit of time taken for data collection and process $c$  &  $[1,13]$  \\	
		\hline
		Duration from finishing data collection to the beginning of the next data collection phase $a$  & $[1,13]$  \\
		\hline
		Unit profit for the performance $\beta$  &  $\left\{ 1, 5 \right\}$  \\	
		\hline		
		Maximum tolerant AoI $K$ &  $200\: \rm{s}$\\	
		\hline		
		Maximum tolerant service latency $H$ &  $50\: \rm{s}$\\
		\hline	
	\end{tabular}\label{table3}
\end{table}

\textit{Case 3:} This case is much more complex and integrated with the above two cases. Considering $U_{s,N,\text{PT}} \geq \cdots \geq U_{s,b,\text{PT}} \geq \cdots \geq U_{s, m+1,\text{PT}} \geq U_{\text{ref}} \geq U_{s,m,\text{PT}} \geq \cdots \geq U_{s,a,\text{PT}} \geq \cdots \geq U_{s,1,\text{PT}}$, the subjective utility in (\ref{F_PT}) can be rewritten as
\begin{equation}
    \begin{aligned}
    U_{s,\text{PT}}  &= MQ \bigg(\eta  \sum_{a = 1}^m U_{s,a,\text{PT}} +  \sum_{b = m+1}^N U_{s,b,\text{PT}}-\eta  m U_{\text{ref}}\\
    &\quad\quad\quad\:\:\:- (N-m)U_{\text{ref}}\bigg) \\
    & = MQ\bigg(\sum_{a = 1}^m \eta  G_a(f_a) -  \sum_{a = 1}^m  d_a f_a + \sum_{b = m+1}^N  G_b(f_b)\\
    &\quad\qquad\:\:\:-  \sum_{b = m+1}^N  d_b f_b -\eta  m U_{\text{ref}}- (N-m)U_{\text{ref}}\bigg),
    \end{aligned}
\end{equation}
where $d_a =  \frac{m-a+1}{\gamma_a}  - \frac{m-a}{\gamma_{a+1}}   +  \frac{N-m}{\gamma_a} - \frac{N-m}{\gamma_{a+1}}  $, $d_b = \frac{N-b+1}{\gamma_b} - \frac{N-b}{\gamma_{b+1}}$ with $m+1 \leq b < N$, and $d_b = \frac{1}{\gamma_b}$ with $b = N$.

For $1 \leq a \leq m$, if $\partial^2 U_{s,\text{PT} }/\partial f_a^2  < 0$, we will use the first-order optimality condition  $\partial U_{s, \text{PT}}/\partial f_a = 0$ and obtain the optimal PT-based solution, i.e., $\hat{f}_{a,\text{PT}}^{*}$. 
We simultaneously consider the lower bound of the update frequency $f_{ \text{min} }$ to derive the PT-based solution of $f_a$, which is given by
\begin{equation}\label{f_a,PT_min}
    \begin{aligned}
    f_{a,\text{PT}}^{*} = \max \Big(\hat{f}_{a,\text{PT}}^{*},f_{ \text{min} } \Big).
    \end{aligned}
\end{equation}
Similarly, taking the first-order and second-order derivatives of $U_{s, \text{PT}}$ concerning $f_{b,\text{PT}}, \:m+1 \leq b \leq N$, the PT-based solution of $f_b$ is given by
\begin{equation}
    \begin{aligned}
    f_{b,\text{PT}}^{*} = f_{b,\text{EUT}}^{*}.
    \end{aligned}
\end{equation}

For \textit{Case 3}, we summarize that
\begin{itemize}
    \item If $\eta = 1$, then $f_{n,\text{PT}}^{*} = f_{n,\text{EUT}}^{*}, \: \forall n $.
   
    \item If $\eta < 1$, then $f_{n,\text{PT}}^{*} \leq f_{n,\text{EUT}}^{*}$ with $1 \leq n \leq m$, and $f_{n, \text{PT} }^{*} = f_{n,\text{EUT}}^{*}$ with $m+1 \leq n \leq N$. Therefore, by seeking $U_{s, m+1,\text{EUT}} \geq U_{\text{ref}} \geq U_{s, m,\text{EUT}}$, the value of $m$ can be confirmed.
   
    \item If $\eta > 1$, then $f_{n,\text{PT}}^{*} \geq f_{n,\text{EUT}}^{*}$ with $1 \leq n \leq m$, and $f_{n,\text{PT}}^{*} = f_{n,\text{EUT}}^{*}$ with $m+1 \leq n \leq N$. When $f_{n,\text{PT}}^{*} \geq f_{m+1,\text{EUT}}^{*}$, the sub-sequences $\small\{f_{n, \text{PT} }^{*}\small\}$ may not follow the essential monotonicity constraint of $f_n$, and they are adjusted by using the scheme in \cite{gao2011spectrum} to meet the demand of \textbf{Lemma \ref{lemma:1}}. Besides, the value of $m$ is determined by using a simple polling method, which consists of the following steps:\\
    \textit{ Step 1: Initialize $m = 1$}.\\
    \textit{ Step 2: Calculate $f_{n,\text{PT}}^{*},\: 1 \leq n \leq m$ based on (\ref{f_a,PT_min}), and let $f_{n,\text{PT}}^{*} = f_{n,\text{EUT}}^{*},\: m+1 \leq n \leq N$.} The scheme in \cite{gao2011spectrum} is used to adjust $\small\{f_{n,\text{PT}}^{*}\small\}$ when necessary.\\
    \textit{ Step 3: Evaluate whether $U_{s, m+1,\text{PT}} \geq U_{\text{ref}} \geq U_{s, m,\text{PT}}$.} If yes, $m$ is confirmed. Otherwise, $m = m+1$.\\
    \textit{ Step 4: Evaluate whether $m < N$.} If yes, go to \textit{Step 2}. Otherwise, the method is terminated.
\end{itemize}

Motivated by the above analysis, the detailed optimal contract design is shown in \textbf{Algorithm \ref{AI_Contract}}. Firstly, the EUT-based solution $f_{n,\text{EUT}}^*$ can be obtained by (\ref{f_k,EUT_min}). Then, based on the above three cases, we compare the sizes of $U_{s,n,\text{EUT}}$ and $U_{\text{ref}}$ sequentially to obtain the optimal PT-based solution $f_{n,\text{PT}}^*$. Finally, by substituting $f_{n,\text{PT}}^*$ into (\ref{pi}), the optimal rewards $R_n^*$ can be calculated. In particular, the computational complexity of \textbf{Algorithm \ref{AI_Contract}} in the worst case is $\mathcal{O}(N(N-1))$, which emphasizes that we can use \textbf{Algorithm \ref{AI_Contract}} to find the optimal contract under PT for all cases that are analyzed above.

\section{Security Analysis and Numerical Results}\label{Results}

In this section, we analyze the security of the cross-chain empowered FL framework and evaluate the performance of the proposed incentive mechanism and the framework. For the simulation setting of the proposed incentive mechanism, we consider $M = 10$ workers and the worker types following the uniform distribution that is distributed in the range of $[0.001,0.01]$. Similar to  \cite{9881813,lim2020information,zhou2021towards,huang2021efficient,ye2022incentivizing}, the main parameters are listed in Table \ref{table3}. 

For evaluating the performance of the proposed incentive mechanism, since \textbf{\emph{Case 2}} and \textbf{\emph{Case 1}} have similar conclusions, we focus our analysis on \textbf{\emph{Case 2}}, i.e., an adjustable update phase and a fixed idle phase. We use MATLAB to conduct experiments and compare the proposed Contract-based incentive mechanism with Asymmetric information (CA) with other incentive mechanisms: i) Contract-based incentive mechanism with Complete information (CC)\cite{hou2017incentive} that the private information of workers (i.e., worker types) is known by the service provider; ii) Contract-based incentive mechanism with Social maximization (CS) \cite{xiong2020multi} that the service provider maximizes social welfare with information asymmetry; iii) Stackelberg Game-based incentive mechanism (SG) with information asymmetry \cite{ye2022incentivizing} that the service provider acting as the leader is not aware of the exact update cost of workers acting as the followers. 

For evaluating the performance of the cross-chain empowered FL framework, we implement this framework by using PySyft based on public datasets\footnote{The public datasets of wisconsin diagnostic breast cancer: \url{ https://goo.gl/U2Uwz2}} of UCI and the Fisco Bcos blockchain with a cross-chain platform named WeCross, which uses the Two-Phase Commit (2PC) protocol as the cross-chain consensus algorithm\cite{kang2022communication}, and the cross-chain system is run on VMware Workstation Pro and the operating system is Ubuntu 22.04 LTS. For FL training, we use Python 3.7.0 running on CPU intel i7-12700 and DDR4 16G RAM to execute tasks on clients A and B. The detailed training process of FL is shown in \textbf{Algorithm \ref{Al_FL}}.

\begin{figure}[t]
	\centering
	\subfloat[ The utility of the service provider.]
	{\includegraphics[width=0.24\textwidth]{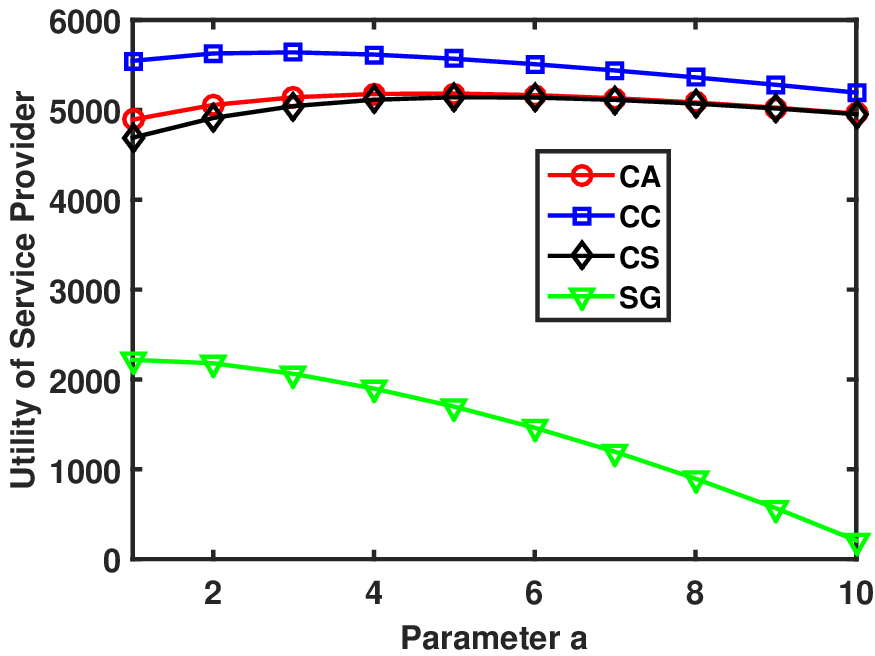}\label{fig:case2_u_s}}
	\subfloat[  Utilities of workers.]
	{\includegraphics[width=0.24\textwidth]{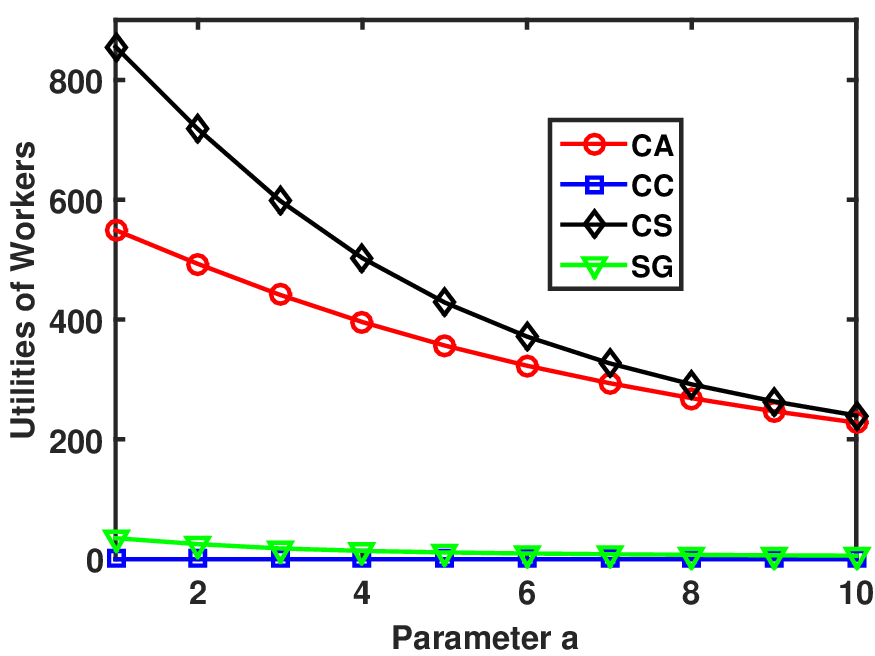}\label{fig:case2_u_n}}
        \captionsetup{font = small}
	\caption{ Utilities of the service provider and workers under different idle duration parameters $a$ in \textbf{\emph{Case 2}}.}\label{fig:U_parameter_a}
\end{figure}
\begin{figure}[t]
	\centering
	\subfloat[ The utility of the service provider.]
	{\includegraphics[width=0.24\textwidth]{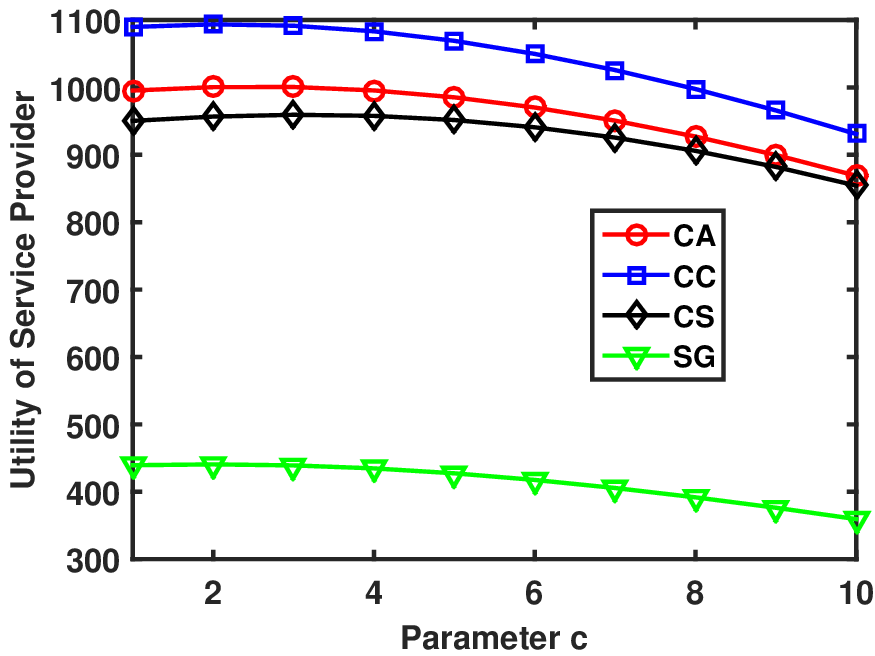}\label{fig:case1_u_s}}
	\subfloat[ Utilities of workers.]
	{\includegraphics[width=0.24\textwidth]{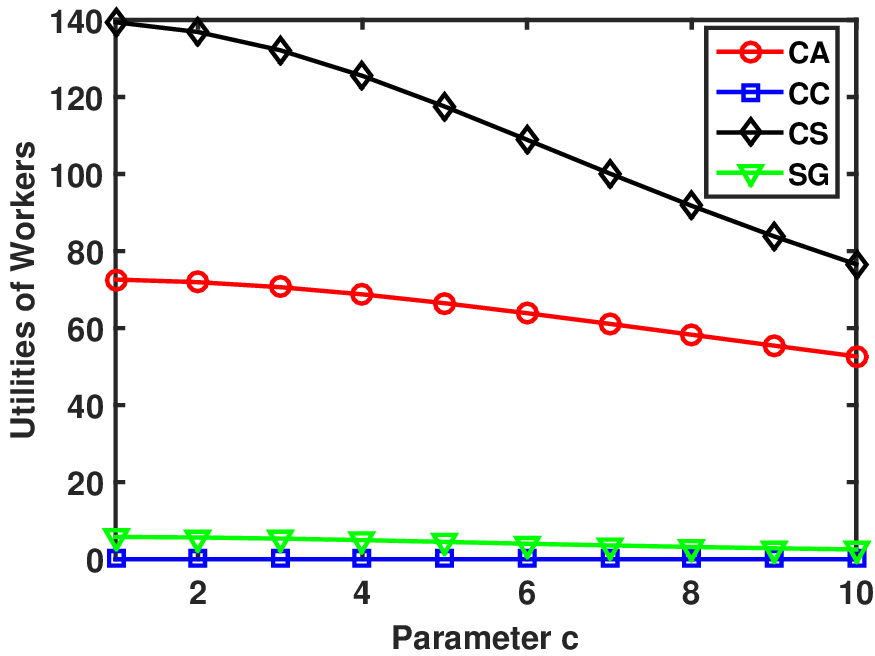}\label{fig:case1_u_n}}
        \captionsetup{font = small}
	\caption{ Utilities of the service provider and workers under different update duration parameters $c$ in \textbf{\emph{Case 1}}.}\label{fig:U_parameter_c}
\end{figure}

\subsection{Security Analysis}
The cross-chain empowered FL framework has the defense ability against many conventional security attacks through blockchain technologies and FL technologies, which satisfies the following security requirements:
\begin{enumerate}[1)]
\item \emph{Privacy protection for users:} With the role of the user-centric privacy-preserving training framework, users can keep sensitive data in the physical space and customize uploading non-sensitive data to virtual spaces for learning-based tasks, thus protecting user privacy effectively.
\item \emph{Without the intervention of the only trusted third party:} Cross-chain interactions are completed in the cross-chain management platform without relying on a third party, thus making the system scalable and robust. Note that the interaction protocol design is based on the 2PC protocol\cite{kang2022communication}. Specifically, the 2PC protocol, as a widely used coordination protocol in distributed systems, can enable secure and efficient cross-chain interactions, which is important in the proposed framework for cross-chain interactions to ensure consistency and reliability.
\item \emph{Data authentication and unforgeability:} We use the Practical Byzantine Fault Tolerance (PBFT) consensus algorithm in the hierarchical cross-chain architecture for lightweight consensus\cite{li2020scalable}. With the role of PBFT, all data are strictly audited and authenticated by delegates (i.e., miners). Besides, because of the decentralized nature of consortium blockchains combined with digitally signed transactions, attackers cannot impersonate users or compromise the system\cite{kang2017enabling}, thus ensuring data unforgeability.
\end{enumerate}

\begin{figure}[t]
	\centering
	\subfloat[ Update frequency.]
	{\includegraphics[width = 0.24\textwidth]{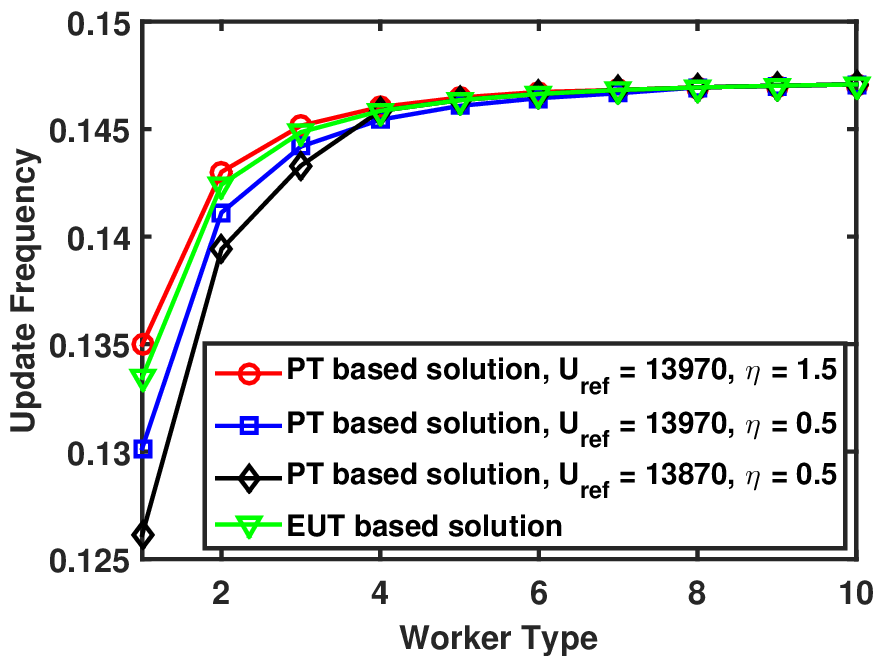}\label{solution_uf_c}}
	\subfloat[  Reward.]
	{\includegraphics[width = 0.24\textwidth]{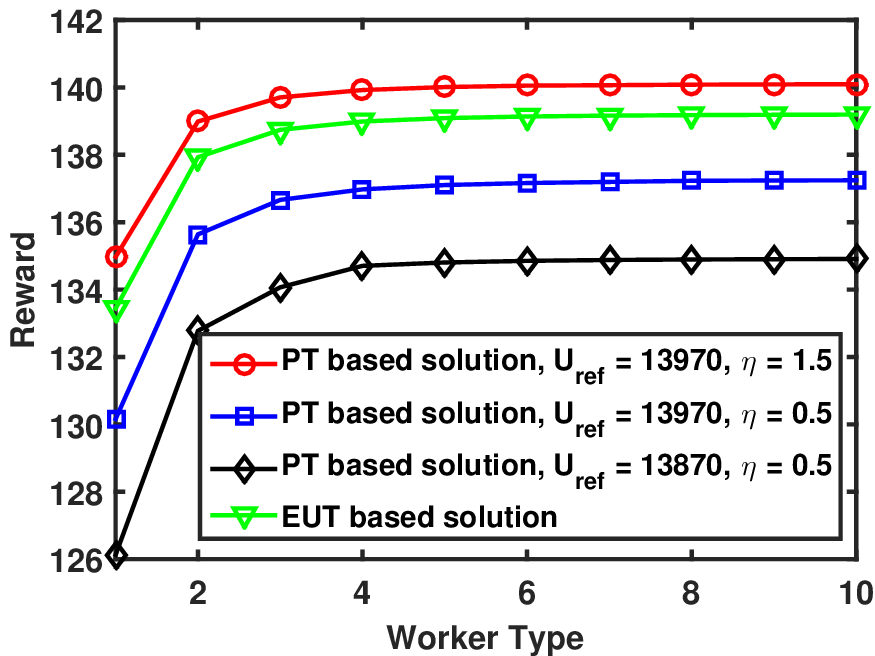}\label{solution_reward_c}}
        \captionsetup{font = small}
	\caption{ Contract items under different preference parameters $U_{\text{ref}}$ and $\eta$ in \textbf{\emph{Case 2}}.}\label{fig:PT_solution_c}
\end{figure}
\begin{figure}[t]
	\centering
	\subfloat[ Update frequency.]
	{\includegraphics[width=0.24\textwidth]{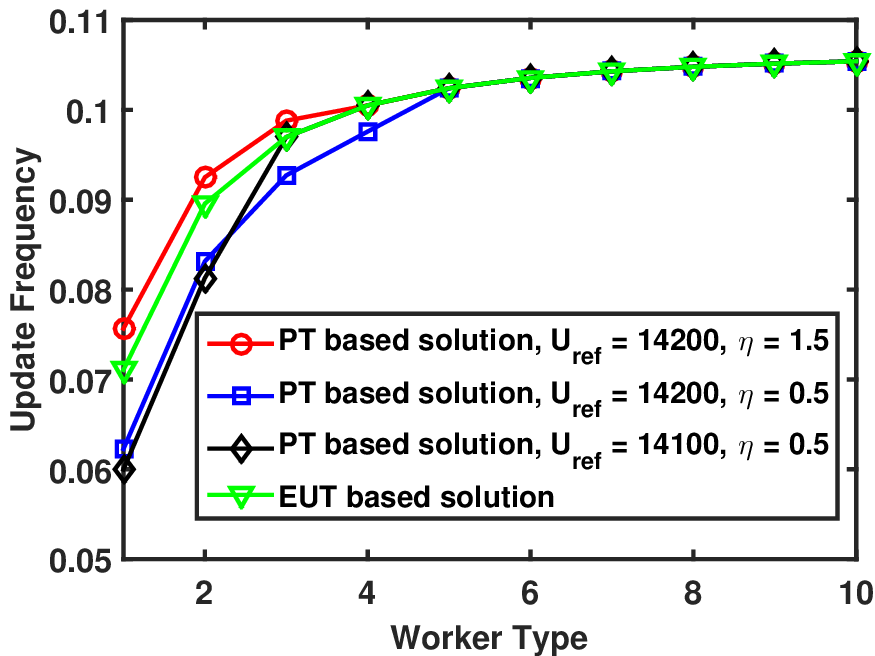}\label{solution_uf_a}}
	\subfloat[  Reward.]
	{\includegraphics[width=0.24\textwidth]{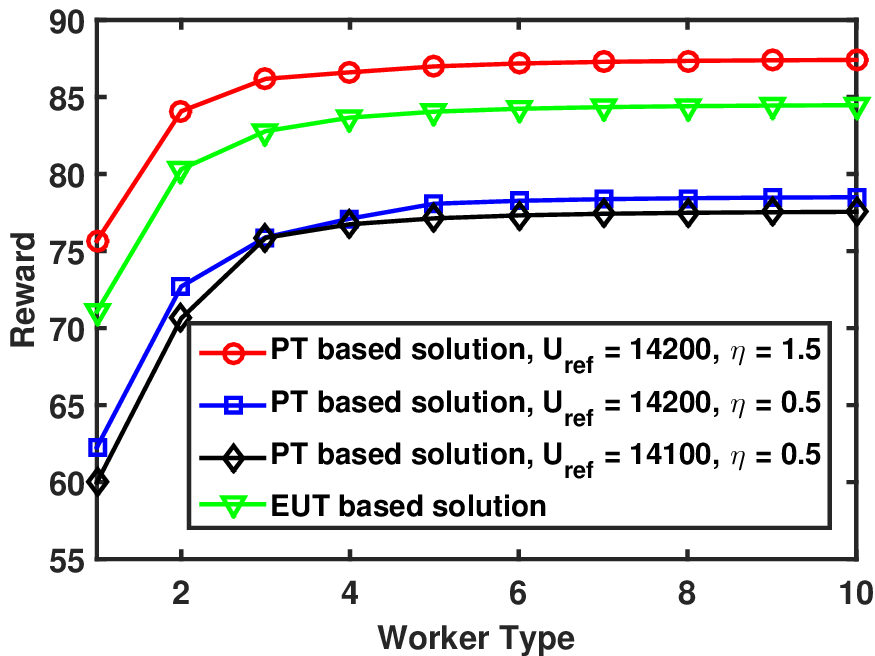}\label{solution_reward_a}}
        \captionsetup{font = small}
	\caption{ Contract items under different preference parameters $U_{\text{ref}}$ and $\eta$ in \textbf{\emph{Case 1}}.}\label{fig:PT_solution_a}
\end{figure}

\begin{figure}[t]
	\centering
	\subfloat[The utility of the service provider.]
	{\label{fig:PTUs_a}\includegraphics[width = 0.24\textwidth]{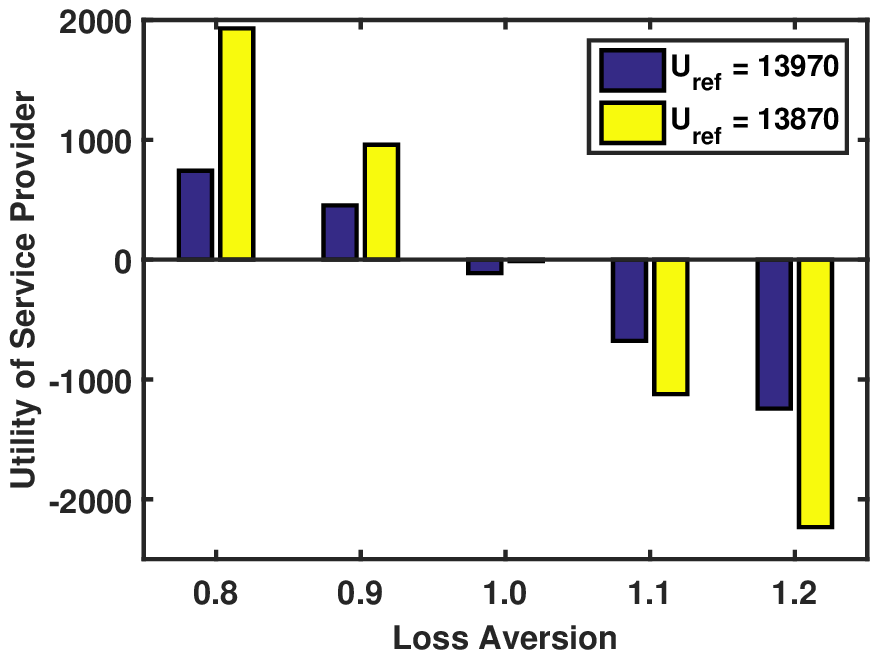}}
	\subfloat[Utilities of workers.]
	{\label{fig:PTUm_a}\includegraphics[width = 0.24\textwidth]{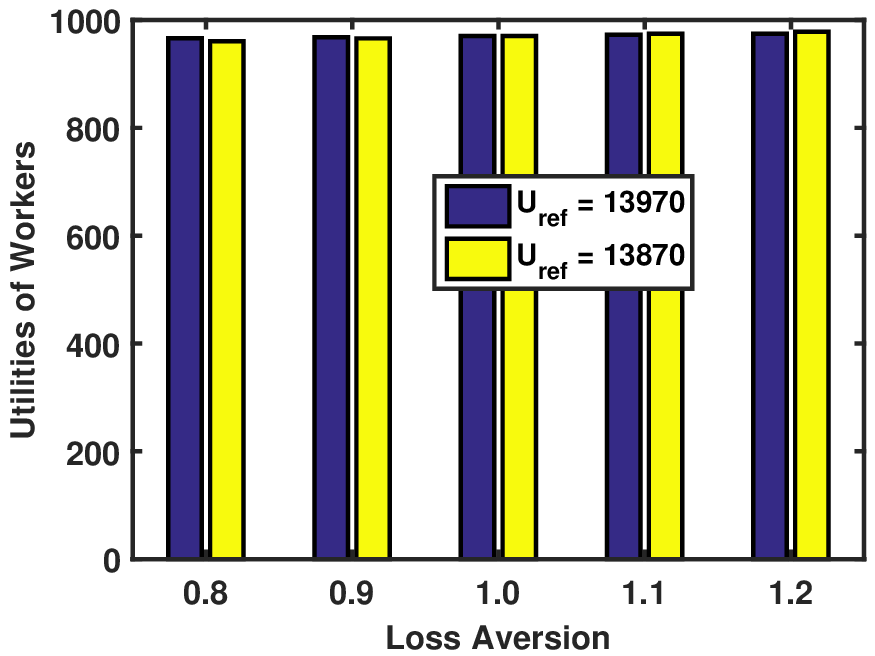}}
        \captionsetup{font = small}
	\caption{Utilities of the service provider and workers under different loss aversions $\eta$ in \textbf{\emph{Case 2}}.}\label{fig:PT_utility_a}
\end{figure}
\begin{figure}[t]
	\centering
	\subfloat[ The utility of the service provider.]
	{\includegraphics[width=0.24\textwidth]{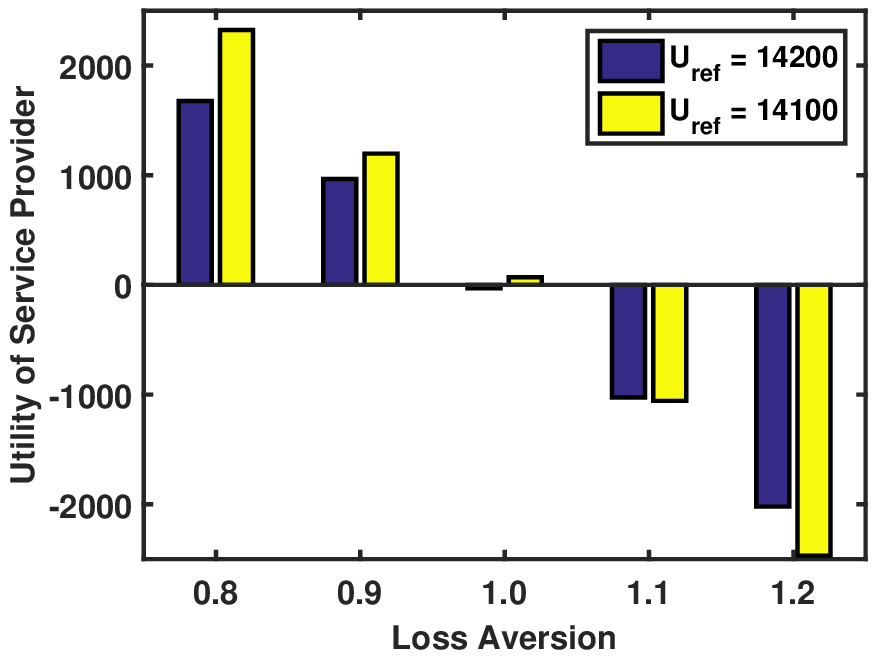}\label{fig:PTUs_c}}
	\subfloat[  Utilities of workers.]
	{\includegraphics[width=0.24\textwidth]{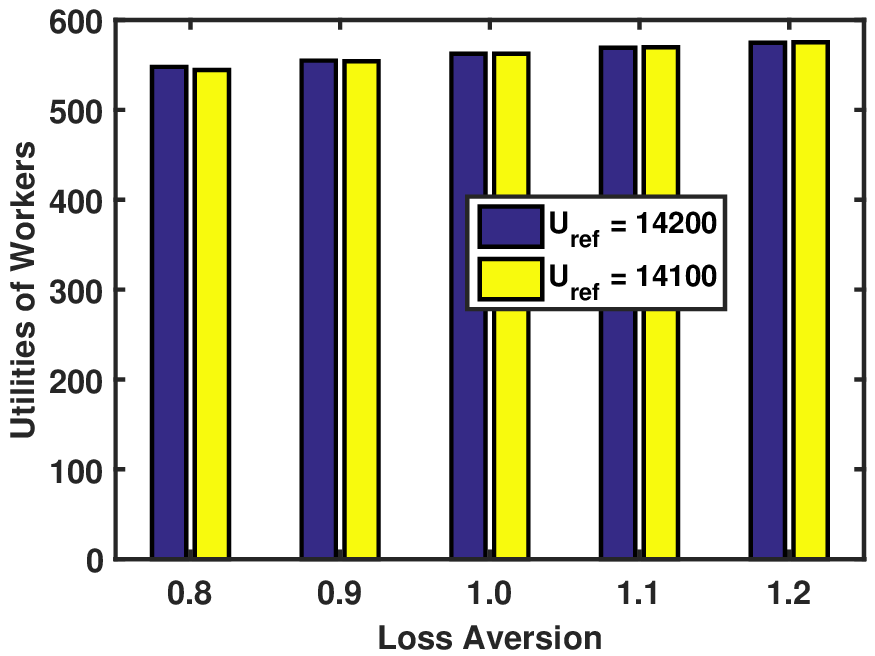}\label{fig:PTUm_c}}
        \captionsetup{font = small}
	\caption{ Utilities of the service provider and workers under different loss aversions $\eta$ in \textbf{\emph{Case 1}}.}\label{fig:PT_utility_c}
\end{figure}

\begin{figure}[t]
	\centering
	\subfloat[ \textbf{\emph{Case 2}}.]
	{\includegraphics[width=0.24\textwidth]{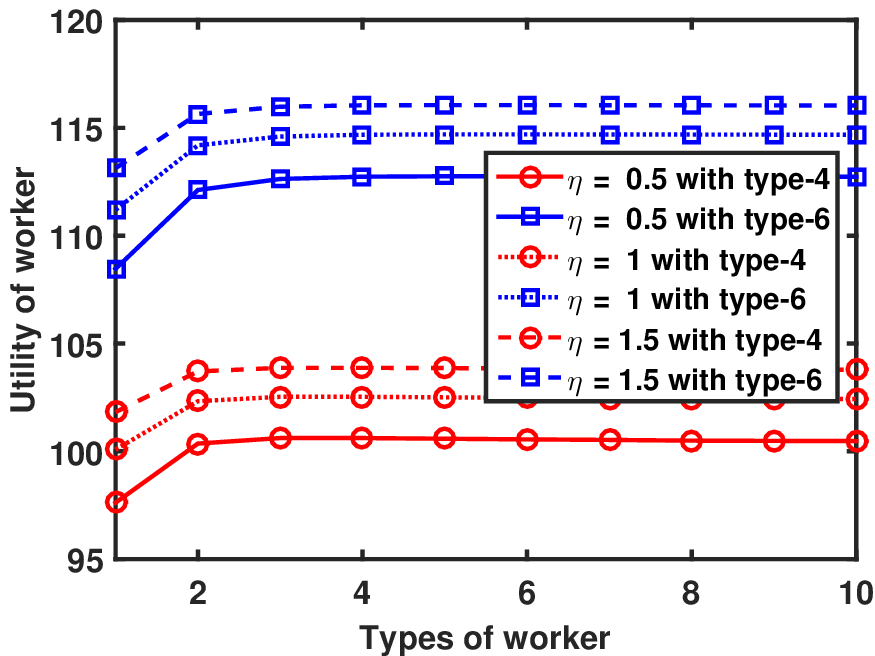}\label{fig:case2_type}}
	\subfloat[ \textbf{\emph{Case 1}}.]
	{\includegraphics[width=0.24\textwidth]{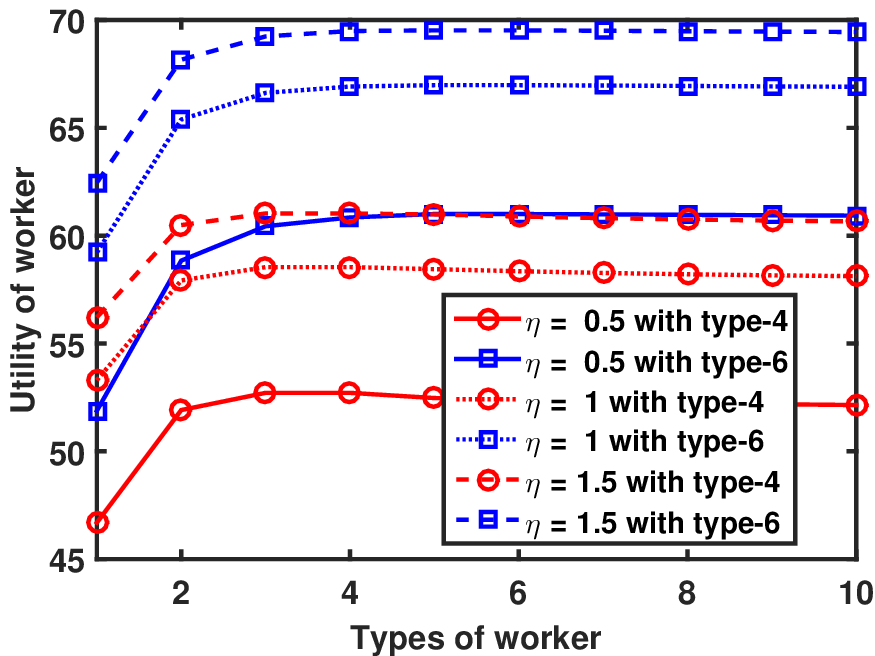}\label{fig:case1_type}}
        \captionsetup{font = small}
	\caption{ Validation of contract properties with the utilities of type-$4$ and type-$6$ workers under different loss aversions.}\label{fig:case_type}
\end{figure}

\begin{figure*}[t]
	\centering
	\subfloat[ Prediction accuracy.]
	{\includegraphics[width=0.24\textwidth]{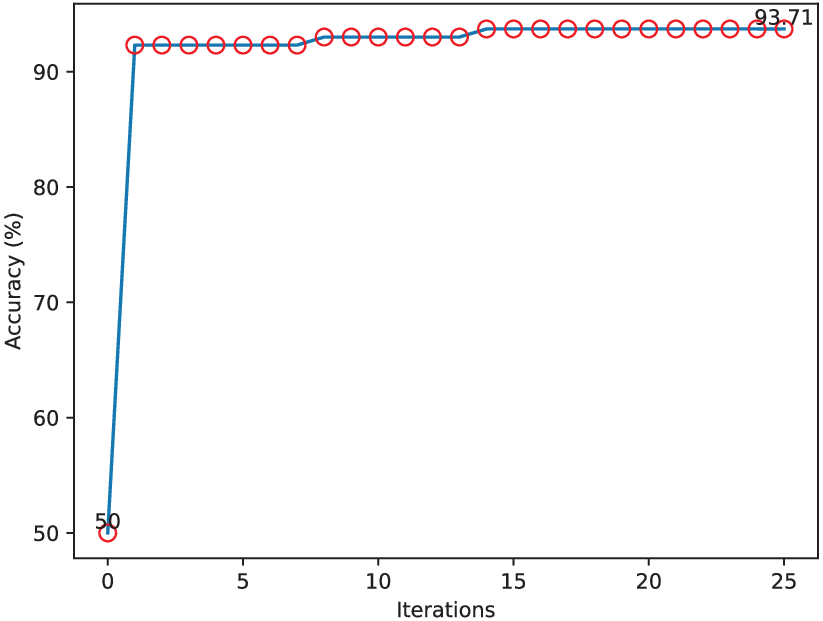}\label{accuracy}}
	\subfloat[ Time spent by vertical FL.]
	{\includegraphics[width=0.24\textwidth]{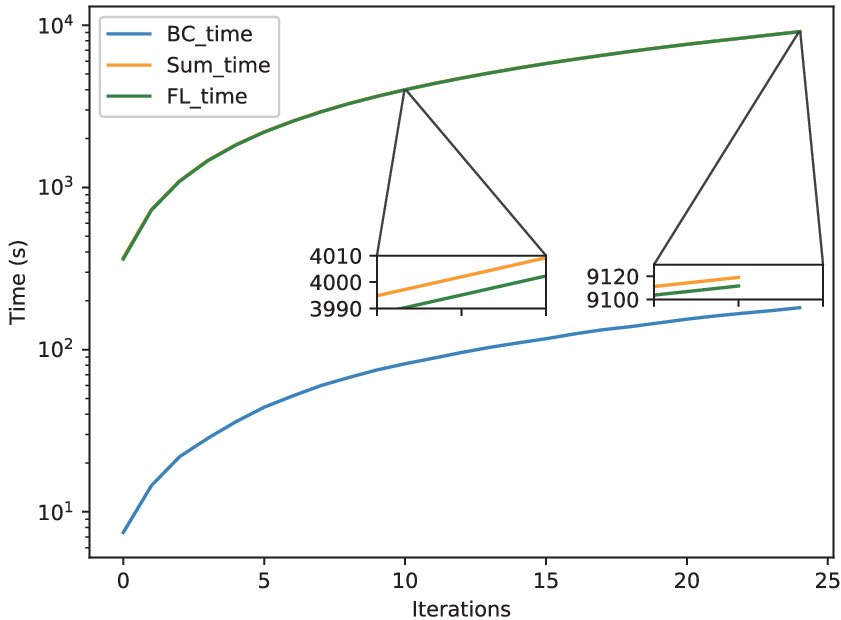}\label{time}}
	\subfloat[ Data storage distribution.]
	{\includegraphics[width=0.24\textwidth]{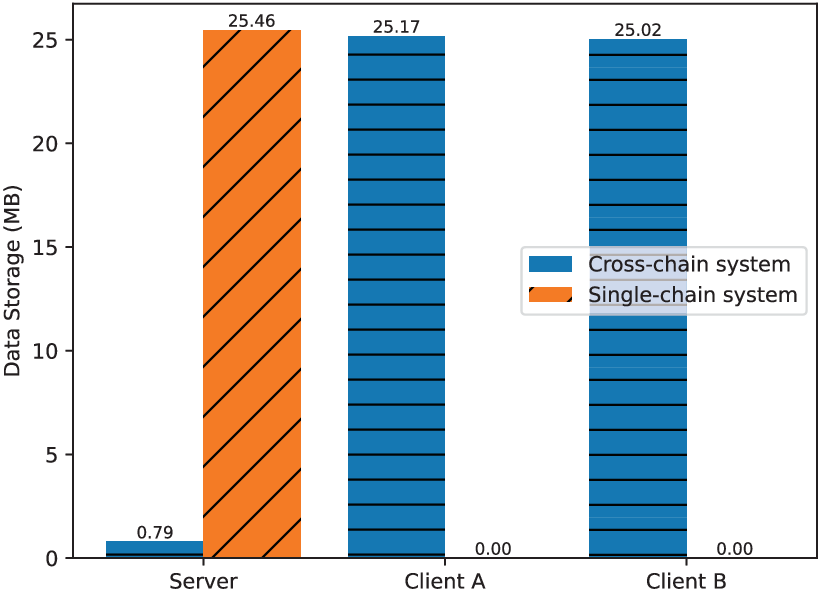}\label{storage}}	
	\subfloat[ Consensus time.]
	{\includegraphics[width=0.24\textwidth]{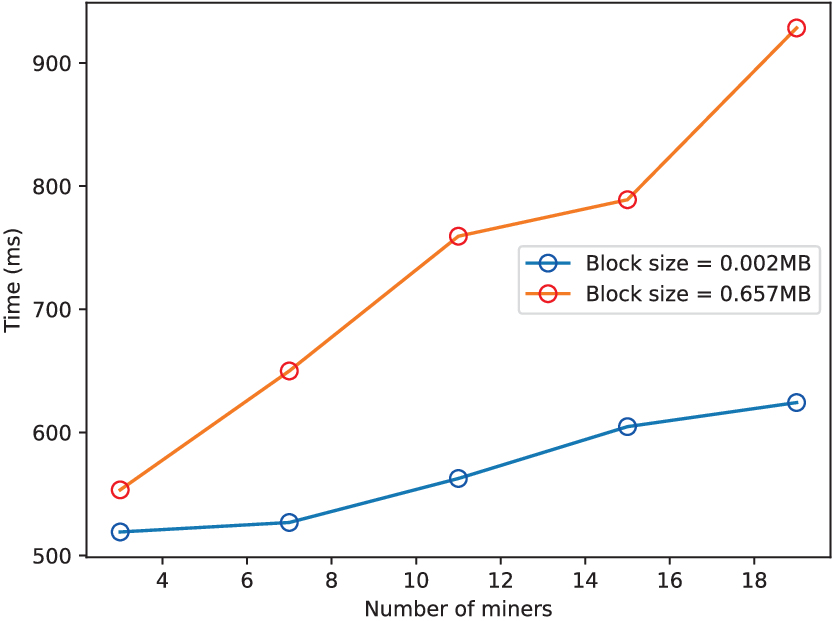}\label{c_time}}
        \captionsetup{font = small}
	\caption{The performance of the cross-chain empowered FL framework.}\label{performance}
\end{figure*}
\begin{figure*}[t]\centering     
\includegraphics[width=18cm, height=13cm]{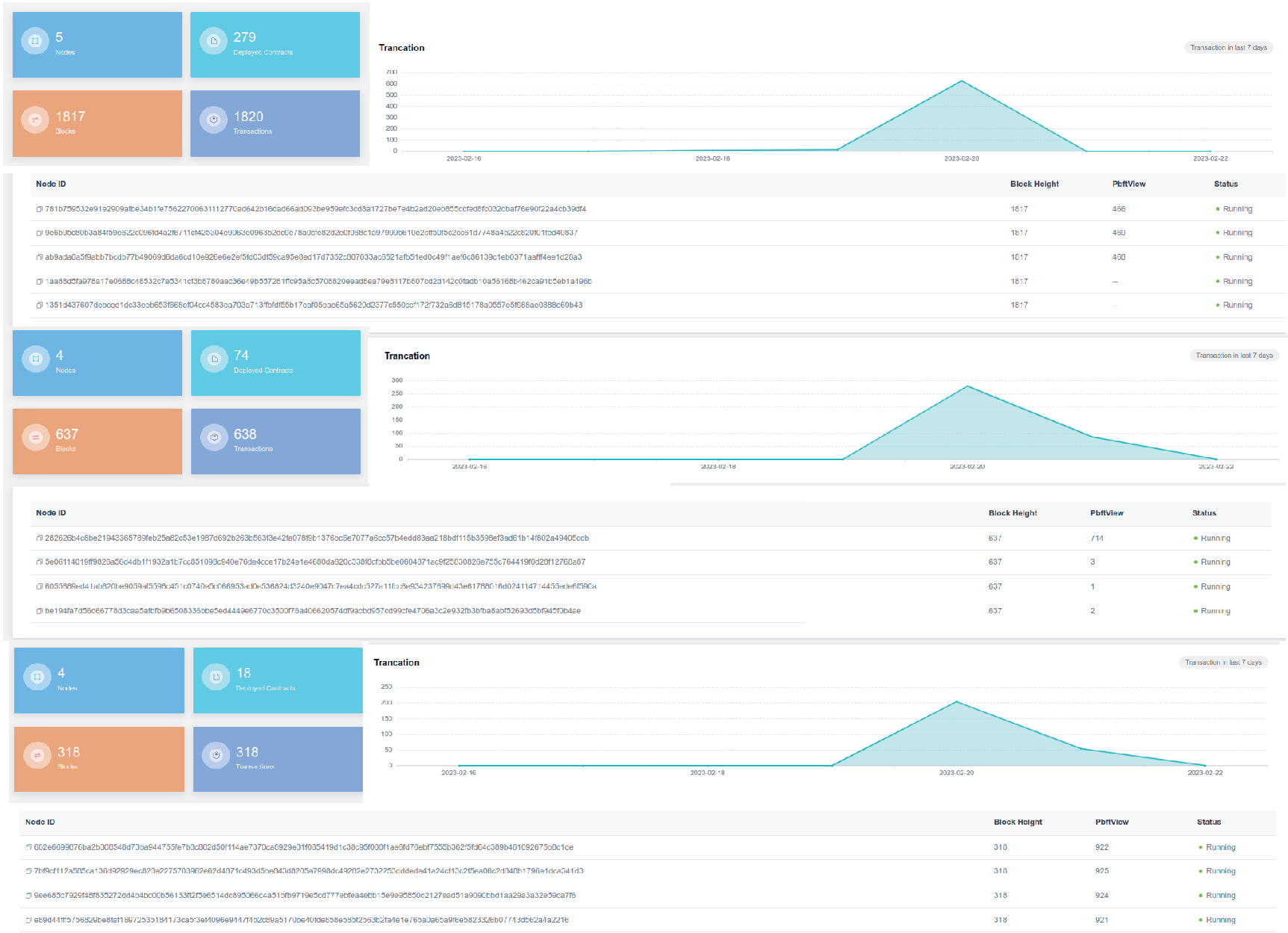}
\captionsetup{font = small}
\caption{Transaction information in three blockchains. Without loss of generality, we set two subchains that are used to store the local models on client A and client B separately and one main chain that is used to store data on the server.}  \label{webase}    
\end{figure*}

\subsection{Performance Analysis of the Proposed Incentive Mechanism}
Figure \ref{fig:U_parameter_a} shows the utilities of the service provider and workers in terms of the idle duration parameter $a$ under different incentive mechanisms. In Fig. \ref{fig:U_parameter_a} (a), with the parameter $a$ increasing, the utility of the service provider first increases and then decreases, which indicates that there exists an optimal parameter $a$ for maximizing the utility of the service provider. For a given parameter $a$, the service provider under the CC mechanism has the highest utility. The reason is that the service provider knows the exact type information of workers and thus offers the most suitable contract item for each worker. Besides, the service provider under the CS mechanism gets a lower utility than that under the proposed CA mechanism. The reason is that the CS mechanism tends to maximize social welfare so that it reaches a balance between the utility of the service provider and that of workers\cite{ye2022incentivizing}. In summary, our proposed CA mechanism allows the service provider to achieve the highest utility under asymmetric information. 
Although the SG mechanism aims at maximizing the objective utilities of both the service provider and workers, the service provider under the SG mechanism gets the lowest utility. The reason is that for all types of workers, the SG mechanism allows only the workers with the highest four types to participate in the FL task under the Nash equilibrium\cite{ye2022incentivizing}. In Fig. \ref{fig:U_parameter_a} (b), with the parameter $a$ increasing, the utilities of workers decrease. We can find that the workers have the highest utility under the CS mechanism and the lowest utility under the CC mechanism. Moreover, the utilities of workers under the SG mechanism are between the utilities of the proposed CA mechanism and the CC mechanism. These results further indicate that our proposed CA mechanism has the highest performance under asymmetric information.


Figure \ref{fig:PT_solution_c} shows the impacts of preference parameters on PT and EUT-based solutions, namely the reference point $U_{\text{ref}}$ and the loss aversion $\eta$. We compare the performance of PT and EUT-based solutions for the proposed CA scheme. The reference point $U_{\text{ref}}$ can affect the PT-based solution of $f_{n}$ for each worker type, which is different from the EUT-based solution of $f_{n}$. As shown in Fig. \ref{fig:PT_solution_c} (a) and Fig. \ref{fig:PT_solution_c} (b), the larger $U_{\text{ref}}$, the more workers with low types improve the subjective utility of the service provider by adjusting their PT-based solution of $f_{n}$. For example, when $U_{\text{ref}} = 13970$ and $\eta = 0.5$, the service provider adjusts the optimal update frequency and the corresponding reward for worker types $1, 2, 3, 4, 5, 6$, and $7$. However, when $U_{\text{ref}}$ is reduced to $ 13870 $ and $\eta$ is unchanged, the service provider only adjusts the optimal update frequency and the corresponding reward for worker types $1, 2, 3, 4, 5$, and $6$. With given a reference point $U_{\text{ref}}$ and a risk-averse behavior (i.e., $\eta>1$), the PT-based solution of each worker is always better than its corresponding EUT-based solution. In turn, with a fixed reference point $U_{\text{ref}}$, when the service provider has a risk-preferred behavior (i.e., $\eta<1$), the PT-based solution of workers with low types (e.g., type-$1$, type-$2$, and type-$3$) is worse than their corresponding EUT-based solutions. 

Figure \ref{fig:PT_utility_a} shows the impacts of preference parameters on the utilities of the service provider and workers. With a given reference point $U_{\text{ref}}$, as the loss aversion $\eta$ increases, the subjective utility of the service provider decreases while the objective utilities of workers increase. The reason is that the increase of the loss aversion $\eta$ means that the service provider tends to have a risk-averse behavior. Therefore, the service provider needs more update frequency from workers with higher types to avoid utility losses, which can increase the objective utilities of workers. Besides, the more update frequency indicates that the service provider needs to send more rewards to  workers, which reduces the subjective utility of the service provider. When the loss aversion $\eta$ is fixed and $\eta < 1$, with the increase of the reference point $U_{\text{ref}}$, the subjective utility of the service provider decreases while the objective utilities of workers increase. In turn, when $\eta > 1$, with the increase of the reference point $U_{\text{ref}}$, the subjective utility of the service provider increases while the objective utilities of workers decrease.

Figure \ref{fig:case_type} shows the validation of contract properties in both two cases. Figure \ref{fig:case_type} (a) shows the utilities of type-$4$ and type-$6$ workers under different loss aversions $\eta$ when selecting all the contract items $(f_{n}, R_{n}),\: n \in \mathcal{N} $ offered by the service provider. In Fig. \ref{fig:case_type} (a), we find that when the service provider has a risk-tolerance behavior (i.e., $\eta = 0.5$), the objective utility of each worker is lower than that when the service provider has a risk-neutral behavior (i.e., $\eta = 1.0$), and the objective utility of each worker is the highest when the service provider has a risk-averse behavior (i.e., $\eta = 1.5$). When the loss aversion $\eta$ is fixed, we can see that each type of worker receives a positive utility when selecting the contract item that fits its type, which demonstrates that our designed contract guarantees the IR condition. Furthermore, each worker can maximize its utility when selecting the contract item that fits its type, which demonstrates that our designed contract guarantees the IC condition. Therefore, we validate that our proposed CA scheme satisfies the IR and IC conditions. Additionally, the utilities of higher types of workers are larger than those of lower types of workers, which demonstrates \textbf{Lemma \ref{lemma:1}}. Based on the above analysis, we can conclude that the service provider can overcome the problem of asymmetric information between the service provider and workers by utilizing the proposed CA scheme.

\subsection{Performance Analysis of the Cross-chain Empowered FL Framework}

Figure \ref{performance} (a) shows the accuracy of FL for the prediction of breast cancer. 
Since users have different numbers and types of features in the dataset, vertical FL training is performed for the prediction of breast cancer. After $25$ iterations, the prediction accuracy can be reached at $93.71\%$, which demonstrates that our proposed cross-chain empowered FL framework has good performance. Figure \ref{performance} (b) shows the time spent by vertical FL in $25$ iterations. The time spent in the blockchain (e.g., consensus time and cross-chain interaction time, etc.) is $181.3\:\rm{s}$, the time spent in the local model training (e.g., the data process and homomorphic encryption, etc.) is $8930.3\:\rm{s}$, and the total time spent in the whole system is $9111.6\:\rm{s}$. Note that the homomorphic encryption algorithm \textit{Paillier} is used for ensuring the security of the training process\cite{fang2021privacy,paillier1999public}. With the iterations increasing, the time consumption of the proposed system increases nonlinearly, and the local model training takes up much time because of homomorphic encryption.

Figure \ref{performance} (c) shows the storage distribution of completing a global iteration in the single-chain system and the cross-chain system. We can see that the total storage of the single-chain system is $25.461\:\rm{MB}$, and the total storage of the cross-chain system is $50.981\:\rm{MB}$, which is the sum of storage on the server, client A, and client B. Although the total storage of the cross-chain system is almost twice as much as the storage of the single-chain system, the storage on the server in the cross-chain system is only $0.787\:\rm{MB}$, which is about $3.09\%$ of the storage in the single-chain system. Therefore, in our cross-chain empowered FL system, the storage pressure on the server is greatly reduced, which allows more clients (i.e., workers) to join FL training. Figure \ref{performance} (d) shows consensus time corresponding to different numbers of miners on the server under different block sizes. From Fig. \ref{performance} (d), we can find that as the number of miners increases, the consensus time increases. Besides, the bigger the block, the more consensus time. Based on the above analysis, the server has higher consensus efficiency due to less stored data in the cross-chain empowered FL system, which indicates the good performance of our proposed system.

\section{Conclusion}\label{Conclusion}
In this paper, we have studied user privacy protection issues and incentive mechanism design for healthcare metaverses. We have proposed a user-centric privacy-preserving framework based on FL technologies for data training in both the virtual space and the physical space of the healthcare metaverse. To ensure secure, decentralized, and privacy-preserving model training, we have designed a decentralized FL architecture based on cross-chain technologies, which consists of a main chain and multiple subchains. Additionally, to improve the service quality of time-sensitive learning tasks, we have applied the AoI as a data-freshness metric and designed an AoI-contract model for incentivizing fresh sensing data sharing in a user-centric manner. Furthermore, we have utilized PT to capture the utility of the service provider considering decision-making under risks and uncertainty. Finally, numerical results demonstrate the effectiveness and reliability of the incentive mechanism and the proposed cross-chain empowered FL framework for healthcare metaverses. For future work, we will further enhance the security and performance of our proposed cross-chain empowered FL framework by considering specific features of health data for healthcare metaverses. Besides, we will use AI tools like deep reinforcement learning or the diffusion model to enhance the solution methodology of the AoI-based contract model under PT.

\bibliographystyle{IEEEtran}
\bibliography{ref}

\begin{thebibliography}{10}
\providecommand{\url}[1]{#1}
\csname url@samestyle\endcsname
\providecommand{\newblock}{\relax}
\providecommand{\bibinfo}[2]{#2}
\providecommand{\BIBentrySTDinterwordspacing}{\spaceskip=0pt\relax}
\providecommand{\BIBentryALTinterwordstretchfactor}{4}
\providecommand{\BIBentryALTinterwordspacing}{\spaceskip=\fontdimen2\font plus
\BIBentryALTinterwordstretchfactor\fontdimen3\font minus
  \fontdimen4\font\relax}
\providecommand{\BIBforeignlanguage}[2]{{%
\expandafter\ifx\csname l@#1\endcsname\relax
\typeout{** WARNING: IEEEtran.bst: No hyphenation pattern has been}%
\typeout{** loaded for the language `#1'. Using the pattern for}%
\typeout{** the default language instead.}%
\else
\language=\csname l@#1\endcsname
\fi
#2}}
\providecommand{\BIBdecl}{\relax}
\BIBdecl

\bibitem{GARAVAND2022101029}
\BIBentryALTinterwordspacing
A.~Garavand and N.~Aslani, ``Metaverse phenomenon and its impact on health: A
  scoping review,'' \emph{Informatics in Medicine Unlocked}, vol.~32, p.
  101029, 2022. [Online]. Available:
  \url{https://www.sciencedirect.com/science/article/pii/S235291482200171X}
\BIBentrySTDinterwordspacing

\bibitem{9880528}
Y.~Wang, Z.~Su, N.~Zhang, R.~Xing, D.~Liu, T.~H. Luan, and X.~Shen, ``A survey
  on metaverse: Fundamentals, security, and privacy,'' \emph{IEEE
  Communications Surveys \& Tutorials}, vol.~25, no.~1, pp. 319--352, 2023.

\bibitem{kurniasih2022digital}
D.~Kurniasih, P.~I. Setyoko, and A.~S. Saputra, ``Digital transformation of
  health quality services in the healthcare industry during disruption and
  society 5.0 era,'' \emph{International Journal of Social and Management
  Studies}, vol.~3, no.~5, pp. 139--143, 2022.

\bibitem{chengoden2022metaverse}
R.~Chengoden, N.~Victor, T.~Huynh-The, G.~Yenduri, R.~H. Jhaveri, M.~Alazab,
  S.~Bhattacharya, P.~Hegde, P.~K.~R. Maddikunta, and T.~R. Gadekallu,
  ``Metaverse for healthcare: A survey on potential applications, challenges
  and future directions,'' \emph{arXiv preprint arXiv:2209.04160}, 2022.

\bibitem{xue2023integration}
H.~Xue, D.~Chen, N.~Zhang, H.-N. Dai, and K.~Yu, ``Integration of blockchain
  and edge computing in internet of things: A survey,'' \emph{Future Generation
  Computer Systems}, vol. 144, pp. 307--326, 2023.

\bibitem{wang2022development}
G.~Wang, A.~Badal, X.~Jia, J.~S. Maltz, K.~Mueller, K.~J. Myers, C.~Niu,
  M.~Vannier, P.~Yan, Z.~Yu \emph{et~al.}, ``Development of metaverse for
  intelligent healthcare,'' \emph{Nature Machine Intelligence}, vol.~4, no.~11,
  pp. 922--929, 2022.

\bibitem{kostick2022nfts}
K.~Kostick-Quenet, K.~D. Mandl, T.~Minssen, I.~G. Cohen, U.~Gasser, I.~Kohane,
  and A.~L. McGuire, ``How nfts could transform health information exchange,''
  \emph{Science}, vol. 375, no. 6580, pp. 500--502, 2022.

\bibitem{zhang2023multi}
T.~Zhang, J.~Shen, C.-F. Lai, S.~Ji, and Y.~Ren, ``Multi-server assisted data
  sharing supporting secure deduplication for metaverse healthcare systems,''
  \emph{Future Generation Computer Systems}, vol. 140, pp. 299--310, 2023.

\bibitem{9881813}
J.~Kang, D.~Ye, J.~Nie, J.~Xiao, X.~Deng, S.~Wang, Z.~Xiong, R.~Yu, and
  D.~Niyato, ``Blockchain-based federated learning for industrial metaverses:
  Incentive scheme with optimal aoi,'' in \emph{2022 IEEE International
  Conference on Blockchain (Blockchain)}, 2022, pp. 71--78.

\bibitem{zhan2020learning}
Y.~Zhan, P.~Li, Z.~Qu, D.~Zeng, and S.~Guo, ``A learning-based incentive
  mechanism for federated learning,'' \emph{IEEE Internet of Things Journal},
  vol.~7, no.~7, pp. 6360--6368, 2020.

\bibitem{8818983}
L.~Zhu, H.~Dong, M.~Shen, and K.~Gai, ``An incentive mechanism using shapley
  value for blockchain-based medical data sharing,'' in \emph{2019 IEEE 5th
  Intl Conference on Big Data Security on Cloud (BigDataSecurity), IEEE Intl
  Conference on High Performance and Smart Computing, (HPSC) and IEEE Intl
  Conference on Intelligent Data and Security (IDS)}, 2019, pp. 113--118.

\bibitem{nie2022blockchain}
X.~Nie, A.~Zhang, J.~Chen, Y.~Qu, and S.~Yu, ``Blockchain-empowered secure and
  privacy-preserving health data sharing in edge-based iomt,'' \emph{Security
  and Communication Networks}, vol. 2022, 2022.

\bibitem{8832210}
J.~Kang, Z.~Xiong, D.~Niyato, S.~Xie, and J.~Zhang, ``Incentive mechanism for
  reliable federated learning: A joint optimization approach to combining
  reputation and contract theory,'' \emph{IEEE Internet of Things Journal},
  vol.~6, no.~6, pp. 10\,700--10\,714, 2019.

\bibitem{wenoptimal}
J.~Wen, X.~Liu, Z.~Xiong, M.~Shen, S.~Wang, Y.~Jiao, J.~Kang, and H.~Li,
  ``Optimal block propagation and incentive mechanism for blockchain networks
  in 6g,'' in \emph{2022 IEEE International Conference on Trust, Security and
  Privacy in Computing and Communications (TrustCom)}, 2022, pp. 369--374.

\bibitem{kahneman2013prospect}
D.~Kahneman and A.~Tversky, ``Prospect theory: An analysis of decision under
  risk,'' in \emph{Handbook of the fundamentals of financial decision making:
  Part I}.\hskip 1em plus 0.5em minus 0.4em\relax World Scientific, 2013, pp.
  99--127.

\bibitem{huang2021efficient}
X.~Huang, R.~Yu, D.~Ye, L.~Shu, and S.~Xie, ``Efficient workload allocation and
  user-centric utility maximization for task scheduling in collaborative
  vehicular edge computing,'' \emph{IEEE Transactions on Vehicular Technology},
  vol.~70, no.~4, pp. 3773--3787, 2021.

\bibitem{8031035}
G.~El~Rahi, S.~R. Etesami, W.~Saad, N.~B. Mandayam, and H.~V. Poor, ``Managing
  price uncertainty in prosumer-centric energy trading: A prospect-theoretic
  stackelberg game approach,'' \emph{IEEE Transactions on Smart Grid}, vol.~10,
  no.~1, pp. 702--713, 2019.

\bibitem{bansal2022healthcare}
G.~Bansal, K.~Rajgopal, V.~Chamola, Z.~Xiong, and D.~Niyato, ``Healthcare in
  metaverse: A survey on current metaverse applications in healthcare,''
  \emph{Ieee Access}, vol.~10, pp. 119\,914--119\,946, 2022.

\bibitem{ali2023metaverse}
S.~Ali, T.~P.~T. Armand, A.~Athar, A.~Hussain, M.~Ali, M.~Yaseen, M.-I. Joo,
  H.-C. Kim \emph{et~al.}, ``Metaverse in healthcare integrated with
  explainable ai and blockchain: Enabling immersiveness, ensuring trust, and
  providing patient data security,'' \emph{Sensors}, vol.~23, no.~2, p. 565,
  2023.

\bibitem{8994206}
J.~Kang, Z.~Xiong, D.~Niyato, Y.~Zou, Y.~Zhang, and M.~Guizani, ``Reliable
  federated learning for mobile networks,'' \emph{IEEE Wireless
  Communications}, vol.~27, no.~2, pp. 72--80, 2020.

\bibitem{zheng2018blockchain}
Z.~Zheng, S.~Xie, H.-N. Dai, X.~Chen, and H.~Wang, ``Blockchain challenges and
  opportunities: A survey,'' \emph{International journal of web and grid
  services}, vol.~14, no.~4, pp. 352--375, 2018.

\bibitem{chang2021blockchain}
Y.~Chang, C.~Fang, W.~Sun \emph{et~al.}, ``A blockchain-based federated
  learning method for smart healthcare,'' \emph{Computational Intelligence and
  Neuroscience}, vol. 2021, 2021.

\bibitem{jatain2022blockchain}
D.~Jatain, V.~Singh, and N.~Dahiya, ``Blockchain base community
  cluster-federated learning for secure aggregation of healthcare data,''
  \emph{Procedia Computer Science}, vol. 215, pp. 752--762, 2022.

\bibitem{wadhwa2022blockchain}
S.~Wadhwa, K.~Saluja, S.~Gupta, and D.~Gupta, ``Blockchain based federated
  learning approach for detection of covid-19 using io mt,'' \emph{Available at
  SSRN 4159195}, 2022.

\bibitem{bolton2004contract}
P.~Bolton and M.~Dewatripont, \emph{Contract theory}.\hskip 1em plus 0.5em
  minus 0.4em\relax MIT press, 2004.

\bibitem{hou2017incentive}
Z.~Hou, H.~Chen, Y.~Li, and B.~Vucetic, ``Incentive mechanism design for
  wireless energy harvesting-based internet of things,'' \emph{IEEE Internet of
  Things Journal}, vol.~5, no.~4, pp. 2620--2632, 2017.

\bibitem{Jinbo}
J.~Wen, J.~Kang, Z.~Xiong, Y.~Zhang, H.~Du, Y.~Jiao, and D.~Niyato, ``Task
  freshness-aware incentive mechanism for vehicle twin migration in vehicular
  metaverses,'' in \emph{IEEE International Conference on Metaverse Computing,
  Networking and Applications (IEEE MetaCom 2023)}.\hskip 1em plus 0.5em minus
  0.4em\relax IEEE, 2023, p. In press.

\bibitem{zhang2022toward}
C.~Zhang, T.~Shen, and F.~Bai, ``Toward secure data sharing for the iot devices
  with limited resources: A smart contract--based quality-driven incentive
  mechanism,'' \emph{IEEE Internet of Things Journal}, 2022.

\bibitem{xuan2020incentive}
S.~Xuan, L.~Zheng, I.~Chung, W.~Wang, D.~Man, X.~Du, W.~Yang, and M.~Guizani,
  ``An incentive mechanism for data sharing based on blockchain with smart
  contracts,'' \emph{Computers \& Electrical Engineering}, vol.~83, p. 106587,
  2020.

\bibitem{karimireddy2022mechanisms}
S.~P. Karimireddy, W.~Guo, and M.~I. Jordan, ``Mechanisms that incentivize data
  sharing in federated learning,'' \emph{arXiv preprint arXiv:2207.04557},
  2022.

\bibitem{kaul2012real}
S.~Kaul, R.~Yates, and M.~Gruteser, ``Real-time status: How often should one
  update?'' in \emph{2012 Proceedings IEEE INFOCOM}.\hskip 1em plus 0.5em minus
  0.4em\relax IEEE, 2012, pp. 2731--2735.

\bibitem{kosta2017age}
A.~Kosta, N.~Pappas, V.~Angelakis \emph{et~al.}, ``Age of information: A new
  concept, metric, and tool,'' \emph{Foundations and Trends{\textregistered} in
  Networking}, vol.~12, no.~3, pp. 162--259, 2017.

\bibitem{zhou2021towards}
X.~Zhou, W.~Wang, N.~U. Hassan, C.~Yuen, and D.~Niyato, ``Towards small aoi and
  low latency via operator content platform: A contract theory-based pricing,''
  \emph{IEEE Transactions on Communications}, vol.~70, no.~1, pp. 366--378,
  2021.

\bibitem{lim2020information}
W.~Y.~B. Lim, Z.~Xiong, J.~Kang, D.~Niyato, C.~Leung, C.~Miao, and X.~Shen,
  ``When information freshness meets service latency in federated learning: A
  task-aware incentive scheme for smart industries,'' \emph{IEEE Transactions
  on Industrial Informatics}, vol.~18, no.~1, pp. 457--466, 2020.

\bibitem{kang2022communication}
J.~Kang, X.~Li, J.~Nie, Y.~Liu, M.~Xu, Z.~Xiong, D.~Niyato, and Q.~Yan,
  ``Communication-efficient and cross-chain empowered federated learning for
  artificial intelligence of things,'' \emph{IEEE Transactions on Network
  Science and Engineering}, vol.~9, no.~5, pp. 2966--2977, 2022.

\bibitem{jin2022towards}
H.~Jin and J.~Xiao, ``Towards trustworthy blockchain systems in the era of
  “internet of value”: development, challenges, and future trends,''
  \emph{Science China Information Sciences}, vol.~65, pp. 1--11, 2022.

\bibitem{jin2021cross}
H.~Jin, X.~Dai, J.~Xiao, B.~Li, H.~Li, and Y.~Zhang, ``Cross-cluster federated
  learning and blockchain for internet of medical things,'' \emph{IEEE Internet
  of Things Journal}, vol.~8, no.~21, pp. 15\,776--15\,784, 2021.

\bibitem{shen2020blockchain}
M.~Shen, H.~Liu, L.~Zhu, K.~Xu, H.~Yu, X.~Du, and M.~Guizani,
  ``Blockchain-assisted secure device authentication for cross-domain
  industrial iot,'' \emph{IEEE Journal on Selected Areas in Communications},
  vol.~38, no.~5, pp. 942--954, 2020.

\bibitem{zhang2018towards}
S.~Zhang, J.~Li, H.~Luo, J.~Gao, L.~Zhao, and X.~S. Shen, ``Towards fresh and
  low-latency content delivery in vehicular networks: An edge caching aspect,''
  in \emph{2018 10th International Conference on Wireless Communications and
  Signal Processing (WCSP)}.\hskip 1em plus 0.5em minus 0.4em\relax IEEE, 2018,
  pp. 1--6.

\bibitem{tang2018multi}
L.~Tang and S.~He, ``Multi-user computation offloading in mobile edge
  computing: A behavioral perspective,'' \emph{IEEE Network}, vol.~32, no.~1,
  pp. 48--53, 2018.

\bibitem{ye2022incentivizing}
D.~Ye, X.~Huang, Y.~Wu, and R.~Yu, ``Incentivizing semisupervised vehicular
  federated learning: A multidimensional contract approach with bounded
  rationality,'' \emph{IEEE Internet of Things Journal}, vol.~9, no.~19, pp.
  18\,573--18\,588, 2022.

\bibitem{jiang2021reliable}
Y.~Jiang, J.~Kang, D.~Niyato, X.~Ge, Z.~Xiong, and C.~Miao, ``Reliable coded
  distributed computing for metaverse services: Coalition formation and
  incentive mechanism design,'' \emph{arXiv preprint arXiv:2111.10548}, 2021.

\bibitem{chen2018efficient}
Y.~Chen, S.~He, F.~Hou, Z.~Shi, and J.~Chen, ``An efficient incentive mechanism
  for device-to-device multicast communication in cellular networks,''
  \emph{IEEE Transactions on Wireless Communications}, vol.~17, no.~12, pp.
  7922--7935, 2018.

\bibitem{gao2011spectrum}
L.~Gao, X.~Wang, Y.~Xu, and Q.~Zhang, ``Spectrum trading in cognitive radio
  networks: A contract-theoretic modeling approach,'' \emph{IEEE Journal on
  Selected Areas in Communications}, vol.~29, no.~4, pp. 843--855, 2011.

\bibitem{paillier1999public}
P.~Paillier, ``Public-key cryptosystems based on composite degree residuosity
  classes,'' in \emph{Advances in Cryptology—EUROCRYPT’99: International
  Conference on the Theory and Application of Cryptographic Techniques Prague,
  Czech Republic, May 2--6, 1999 Proceedings 18}.\hskip 1em plus 0.5em minus
  0.4em\relax Springer, 1999, pp. 223--238.

\bibitem{xiong2020multi}
Z.~Xiong, J.~Kang, D.~Niyato, P.~Wang, H.~V. Poor, and S.~Xie, ``A multi
  dimensional contract approach for data rewarding in mobile networks,''
  \emph{IEEE Transactions on Wireless Communications}, vol.~19, no.~9, pp.
  5779--5793, 2020.

\bibitem{li2020scalable}
W.~Li, C.~Feng, L.~Zhang, H.~Xu, B.~Cao, and M.~A. Imran, ``A scalable
  multi-layer pbft consensus for blockchain,'' \emph{IEEE Transactions on
  Parallel and Distributed Systems}, vol.~32, no.~5, pp. 1146--1160, 2020.

\bibitem{kang2017enabling}
J.~Kang, R.~Yu, X.~Huang, S.~Maharjan, Y.~Zhang, and E.~Hossain, ``Enabling
  localized peer-to-peer electricity trading among plug-in hybrid electric
  vehicles using consortium blockchains,'' \emph{IEEE Transactions on
  Industrial Informatics}, vol.~13, no.~6, pp. 3154--3164, 2017.

\bibitem{fang2021privacy}
H.~Fang and Q.~Qian, ``Privacy preserving machine learning with homomorphic
  encryption and federated learning,'' \emph{Future Internet}, vol.~13, no.~4,
  p.~94, 2021.

\end{thebibliography}
\end{document}